%% file: pagnan_linear.tex
\documentclass[a4paper,10pt]{article}
\mathcode`\<="4268     
\mathcode`\>="5269     
\mathcode`\:="603A     
\mathchardef\gt="313E  
\mathchardef\lt="313C  
\usepackage{mathbbol}
\usepackage{amsthm}
\usepackage{bussproofs}
\usepackage{cmll}
\usepackage{tikz}
\theoremstyle{definition}

\def\pos{\ensuremath{\xymatrix@C=4ex{\ar@{|->}[r]^{+}&}}}
\def\neg{\ensuremath{\xymatrix@C=4ex{\ar@{|->}[r]^{-}&}}}
\def\rll{\ensuremath{\textup{RLL}}}

\def\aa#1#2{\ensuremath{\catebf{A}_{#1#2}}}

\def\ii#1#2{\ensuremath{\catebf{I}_{#1#2}}}

\def\AA#1#2{\ensuremath{#1\arr#2}}
\def\EE#1#2{\ensuremath{#1\arr\b\leftarrow#2}}
\def\II#1#2{\ensuremath{#1\leftarrow\b\arr#2}}
\def\OO#1#2{\ensuremath{#1\leftarrow\b\arr\b\leftarrow#2}}
\def\Ao#1#2{\ensuremath{#2\leftarrow#1}}
\def\Eo#1#2{\ensuremath{#2\arr\b\leftarrow#1}}
\def\Io#1#2{\ensuremath{#2\leftarrow\b\arr#1}}
\def\Oo#1#2{\ensuremath{#2\arr\b\leftarrow\b\arr#1}}

\newtheorem{theorem}{Theorem}[section]
\newtheorem{lemma}[theorem]{Lemma}
\newtheorem{proposition}[theorem]{Proposition}

\newtheorem{definition}[theorem]{Definition}
\newtheorem{remark}[theorem]{Remark}
\newtheorem{notation}[theorem]{Notation}
\newtheorem{example}[theorem]{Example}
\newtheorem{examples}[theorem]{Examples}
\newtheorem{terminology}[theorem]{Terminology}

\usepackage[english]{babel}
\usepackage[all]{xy}
\usepackage{txfonts}
\usepackage{dsfont}
\usepackage{bbm}
\usepackage{latexsym}

\title{Syllogisms in Rudimentary Linear Logic, Diagrammatically
\footnote{The final publication is available at springerlink.com http://link.springer.com/article/10.1007\%2Fs10849-012-9170-4\#page-1}}

\author{Ruggero Pagnan\\
DISI, University of Genova, Italy\\
\texttt{ruggero.pagnan@disi.unige.it}}
\date{}
\begin{document}
\maketitle
\begin{abstract}
We present a reading of the traditional syllogistics in a fragment of the propositional intuitionistic multiplicative linear logic and prove that with respect to a diagrammatic logical calculus that we introduced in a previous paper, a syllogism is provable in such a fragment if and only if it is diagrammatically provable. We extend this result to syllogistics with complemented terms \`a la De Morgan, with respect to a suitable extension of the diagrammatic reasoning system for the traditional case and a corresponding reading of such De Morgan style syllogistics in the previously referred to fragment of linear logic.
\end{abstract}

\input{macros}
\input{intro}

\input{contextualization}

\input{motivation}

\input{prelim}

\input{syllorll}
\input{catform}

\input{redrules}
\input{demorgan}

\input{final}
\input{conclusions}

\bibliographystyle{plain}
\bibliography{BiblioTeX}

\end{document}

%% file: macros.tex
\def\a{\ensuremath{(a)}}
\def\adj#1#2{\ensuremath{\xymatrix@C.2cm{#1\ar@{-|}[r]&#2}}}
%
\def\adjpair#1#2#3#4#5{\ensuremath{\xymatrix{#1\ar@/^/[r]^{#3}_{\hole}="1"&#2\ar@/^/[l]^{#4}_{\hole}="2" \ar@{} "2";"1"|(.3){#5}}}}
%
\def\and{\ensuremath{\wedge}}
%
\def\arr{\ensuremath{\rightarrow}}
%
\def\Arr{\ensuremath{\Rightarrow}}
%
%
%
\def\arrow#1{\ensuremath{\cateb{#1}^{\arr}}}
%
\def\as{\ensuremath{\ast}}
%
%
\def\b{\ensuremath{\bullet}}
%
\def\BC{\ensuremath{(BC)}}
%
\def\BCd{\ensuremath{(BC)_d}}
\def\beps{\ensuremath{\backepsilon}}
%
\def\bicat#1{\ensuremath{\mathcal{#1}}}
%
\def\bf#1{\ensuremath{\mathbf{#1}}}
%
\def\bpararrow#1#2#3#4{\ensuremath{\xymatrix{#1\ar@<.9ex>[rr]^{#3}\ar@<-.9ex>[rr]_{#4}&&#2}}}
%
\def\bsquare#1#2#3#4#5#6#7#8{\xymatrix{#1\ar[dd]_{#8}\ar[rr]^#5&&#2\ar[dd]^{#6}\\
\\
#3\ar[rr]_{#7}&&#4}}
\def\bu{\ensuremath{\bullet}}
\def\btoind#1{#1\index{#1@\textbf{#1}}}
%
\def\card#1{\ensuremath{|#1|}}
\def\cate#1{\ensuremath{\mathcal{#1}}}
\def\cateb#1{\ensuremath{\mathbbm{#1}}}
%
\def\catebf#1{\ensuremath{\mathbf{#1}}}
\def\catepz#1{\ensuremath{\mathpzc{#1}}}
%
\def\cateop#1{\ensuremath{\cate#1^{op}}}
\def\comparrow#1#2#3#4#5#6{\ensuremath{\xymatrix{#6#1\ar[r]#4&#2\ar[r]#5&#3}}}
\def\cons#1{\ensuremath{\newdir{|>}{%
!/4.5pt/@{|}}\xymatrix@1@C=.3cm{\ar@{|=}[r]_{#1}&}}}
\def\D{\ensuremath{(D)}}
\def\enunciato#1#2#3#4#5#6#7{\newtheorem{#4}#5{#1}#6                                              
\begin{#4}#2\label{#7}
#3
\end{#4}}
%
\def\dfn#1#2#3#4#5#6#7{\newtheorem{#4}#5{#1}#6                                              
\begin{#4}#2\label{#7}
\emph{#3}
\end{#4}}
%
\def\epiarrow#1#2#3#4{\ensuremath{\newdir{|>}{%
!/4.5pt/@{|}*:(1,-.2)@^{>}*:(1,+.2)@_{>}}\xymatrix{#3#1\ar@{-|>}[r]#4&#2}}}
\def\epitip{\ensuremath{\newdir{|>}{%
!/4.5pt/@{|}*:(1,-.2)@^{>}*:(1,+.2)@_{>}}}}
\def\eps#1{\ensuremath{\epsilon#1}}
\def\esempio#1#2#3#4#5#6#7{\newtheorem{#4}#5{#1}#6                                              
\begin{#4}#2\label{#7}
\textup{#3}
\begin{flushright}
$\blacksquare$
\end{flushright}
\end{#4}}
%
\def\et#1{\ensuremath{\eta#1}}
\def\Ex#1{\mbox{\boldmath\ensuremath{\exists}}{#1}}
\def\ex#1{\ensuremath{\exists#1}}
\def\F{\ensuremath{(F)}}
\def\farrow#1#2#3{\ensuremath{\xymatrix{#3:#1\ar[r]&#2}}}
\def\Fi{\ensuremath{\Phi}}
%
\def\fib#1#2#3{\ensuremath{
\catebf{#3}:\cateb{#1}\arr\cateb{#2}}}
%
%
\def\fibs#1#2#3#4{\ensuremath{\begin{array}{ll}
\slice{\catebf{#1}}{#2}\\
\,\downarrow^{\catebf{#4}}\\
\catebf{#3}
\end{array}}}
\def\fibss#1#2#3#4#5{\ensuremath{\begin{array}{ll}
\slice{\catebf{#1}}{#2}\\
\,\downarrow^{#5}\\
\slice{\catebf{#3}}{#4}
\end{array}}}
\def\forev#1{\ensuremath{\forall~#1,~\forall~}}
\def\fract#1{\ensuremath{\cateb{#1}[\Sigma^{-1}]}}
%
\def\G#1{\ensuremath{\int{#1}}}
%
%
\def\harrow{\ensuremath{\rightharpoonup}}
%
\def\hook{\ensuremath{\hookrightarrow}}
\def\hset#1#2#3{\ensuremath{\cate{#1}(#2,#3)}} 
\def\Im#1{\ensuremath{\mathit{Im}#1}}
\def\implies{\ensuremath{\Rightarrow}}
\def\indprod#1#2{\ensuremath{\prod#1#2}}
%
\def\infer#1#2{\ensuremath{\left.\begin{array}{cc}
#1\\
\xymatrix{\ar@{-}@<-.3ex>[rrrr]\ar@{-}@<.3ex>[rrrr]&&&&}\\
#2\\
\end{array}\right.}}
%
%
\def\it#1{\ensuremath{\mathit{#1}}}
%
\def\l#1{\ensuremath{#1_{\bot}}}
\def\larrow{\ensuremath{\leftarrow}}
\def\lcomparrow#1#2#3#4#5{\ensuremath{\xymatrix{#1\ar[rr]#4&&#2\ar[rr]#5&&#3}}}
\def\lfract#1{\ensuremath{[\Sigma^{-1}]\cateb{#1}}}
%
\def\lhookarrow#1#2#3{\ensuremath{\xymatrix{#1\ar@{^{(}->}[rr]^{#3}&&#2}}}
\def\mod#1#2{\ensuremath{\cateb{Mod}(#1,#2)}}
\def\Mono#1{\ensuremath{\mathit{Mono}#1}}
\def\monoarrow#1#2#3#4{\ensuremath{\newdir{
>}{{}*!/-5pt/@{>}}\xymatrix{#3#1\ar@{ >->}[r]#4&#2}}}
\def\monotail{\ensuremath{\newdir{
>}{{}*!/-5pt/@{>}}}}
\def\N{\ensuremath{(N)}}
\def\natarrow#1#2#3{\ensuremath{\newdir{|>}{%
!/4.5pt/@{|}*:(1,-.2)@^{>}*:(1,+.2)@_{>}}\xymatrix{#3:#1\ar@{=|>}[r]&#2}}}
\def\natbody{\ensuremath{\newdir{|>}{%
!/4.5pt/@{|}*:(1,-.2)@^{>}*:(1,+.2)@_{>}}}}
\makeatletter
\def\natpararrow#1#2#3#4#5{\@ifnextchar(
 {\Natpararrow{#1}{#2}{#3}{#4}{#5}}
 {\Natpararrow{#1}{#2}{#3}{#4}{#5}(.5)}}
\makeatother
\def\Natpararrow#1#2#3#4#5(#6){\ensuremath{\newdir{|>}{%
!/4.5pt/@{|}*:(1,-.2)@^{>}*:(1,+.2)@_{>}}
\xymatrix{#1\ar@<1.5ex>[rr]^(#6){#3}|(#6){\vrule height0pt depth.7ex
width0pt}="a"\ar@<-1.5ex>[rr]_(#6){#4}|(#6){\vrule height1.5ex width0pt}="b"&&#2
\ar@{=|>} "a";"b"#5}}}
\def\om{\ensuremath{\omega}}
%
\def\ov#1{\ensuremath{\overline{#1}}}
%
\def\ovarr#1{\ensuremath{\overrightarrow{#1}}}
%
\def\pararrow#1#2#3#4{\ensuremath{\xymatrix{#1\ar@<.8ex>[r]^{#3}\ar@<-.8ex>[r]_{#4}&#2}}}
\def\p#1{\ensuremath{\mathds{P}_{\catepz{S}}}(#1)}
\def\P#1{\ensuremath{\mathnormal{P}(#1)}}
\def\pr#1{\textbf{Proof:~}#1 
\begin{flushright}
$\Box$
\end{flushright}
}
%
%
%
%
\def\rel#1#2#3{\ensuremath{\xymatrix{#1:#2\ar[r]
|-{\SelectTips{cm}{}\object@{|}} &#3}}}
%
%
\def\res#1#2{\ensuremath{#1_{\rbag #2}}}
\def\S#1{\ensuremath{\catepz{S}_{\cateb{#1}}}}
\def\sectoind#1#2{#1\index{#2!#1}}
%
\def\sinfer#1#2{\ensuremath{\left.\begin{array}{cc}
#1\\
\xymatrix{\ar@{-}[rrrr]&&&&}\\
#2\\
\end{array}\right.}}
%
\def\slice#1#2{\ensuremath{#1/{\textstyle#2}}}
\def\sslice#1#2{\ensuremath{#1//{\textstyle#2}}}
\def\square#1#2#3#4#5#6#7#8{$$\xymatrix{#1\ar[d]_{#8}\ar[r]^{#5}&#2\ar[d]^{#6}\\
#3\ar[r]_{#7}&#4}$$}
\def\sub#1#2#3{\ensuremath{#1(#2#3)~}}
\def\summ#1{\ensuremath{\sum#1}}
%
\def\ti#1{\ensuremath{\tilde#1}}
%
\def\tens{\ensuremath{\otimes}}
%
\def\teo#1{\ensuremath{\tau#1}}
\def\ter#1{\ensuremath{\tau_{\cate{#1}}}}
\def\teta{\ensuremath{\theta}}
%
%
\def\toind#1{#1\index{#1}}
%
\def\un#1{\underline{#1}}
\def\vcell{\ensuremath{\ar@<-1ex>@{}[r]_{\hole}="a"\ar@<+1ex>@{}[r]^{\hole}="b"}}
%
\def\veps{\ensuremath{\varepsilon}}
%
\def\vfi{\ensuremath{\varphi}}
\def\w{\ensuremath{\wedge}}
\def\wti#1{\ensuremath{\widetilde#1}}
\def\y#1{\ensuremath{y#1}}

%% file: intro.tex
\section{Introduction}
In this paper we investigate the existing connections between the diagrammatic calculus for the traditional Aristotelian syllogistic that we introduced in~\cite{Pagnan2012-PAGADC} as a logical formal system $\textup{SYLL}^+$, and a reading of the traditional Aristotelian syllogistic itself in a minimal extension $\rll^{\bot}$ of the propositional multiplicative fragment of intuitionistic linear logic which is \emph{rudimentary linear logic}, see~\cite{MR1157804}. So, this paper represents a further attempt towards the building of a bridge between the frameworks of diagrammatic logical reasoning and of linguistic logical reasoning, in their generality. The specific nature of the logical systems we will be dealing with makes things interesting. On one side, the diagrammatic system at issue is tailored to syllogistics and has remarkable features. On the other side, linear logic proved itself very useful in providing new insight into well established contexts of knowledge. Hopefully, the investigation of the interaction between the two would contribute to the maintainance of a lively interest in the study and deeper understanding of syllogistic reasoning.

In the multiplicative fragment of classical linear logic, proofs in the corresponding sequent calculus are conveniently represented graphically as \emph{proof-nets}, capturing their essential geometric content. In~\cite{Abrusci}, the traditional syllogistic is investigated in the classical multiplicative linear logic through a geometric analysis of the Aristotelian reduction rules, via proof-nets. Our reading of the Aristotelian syllogistics within $\rll^{\bot}$ is based on four propositional formulas that we named \emph{categorical formulas}, which suitably translate the corresponding four Aristotelian categorical propositions. 
In~\cite{Abrusci}, such categorical formulas can also be found but we came to study them independently. Our diagrammatic calculus of syllogisms captures the essential geometric content of the corresponding intuitionistic proofs in the sequent calculus of $\rll^{\bot}$.

Our approach extends to the 19th century De Morgan's syllogistics. For this, an extension of $\textup{SYLL}^+$ has to be considered, which was not in our previous paper, together with suitable \emph{new categorical formulas}.
As far as we know, the new categorical formulas and the linear logic reading of De Morgan's syllogistics by means of them is completely new.

An extended contextualization of the paper with respect to the framework of diagrammatic reasoning is dealt with in section~\ref{context}, which contains an overview of the most prominent diagrammatic reasoning systems, with a focus on those founded on Euler diagrams in particular. An extended motivation of the paper with respect to the contexts of linear logic and some recent contributions in the field of syllogistic reasoning with diagrams is dealt with in section~\ref{motivation}.
\subsection{Structure of the paper}
In section~\ref{syllru} we recall the basics of the traditional syllogistic and connect it to rudimentary linear logic in section~\ref{syllrll}, where we prove that a (strengthened) syllogism is valid if and only if the sequent which corresponds to it is provable in the sequent calculus of $\rll^{\bot}$.
In section~\ref{sec4} the systems $\rll^{\bot}$ and $\textup{SYLL}^{+}$ are put in direct connection and we prove that a (strengthened) syllogism is provable in the first if and only if it is provable in the second.
In section~\ref{redrules}, the Aristotelian reduction rules for the syllogistics are briefly discussed. They were not dealt with in~\cite{Pagnan2012-PAGADC}.
In section~\ref{demorgan}, the modern 19th century De Morgan style syllogistics with complemented terms is considered, which was not in~\cite{Pagnan2012-PAGADC}. A suitable reading of it within $\rll^{\bot}$ is considered too.
Furthermore, an extension $\textup{SYLL}^{+*}$ of $\textup{SYLL}^{+}$ which was not introduced in loc. cit is here introduced. We prove that a De Morgan syllogism is valid if and only if it is diagrammatically provable in $\textup{SYLL}^{+*}$ and that a De Morgan syllogism is provable in $\rll^{\bot}$ if and only if it is diagrammatically provable in $\textup{SYLL}^{+*}$.
In section~\ref{comm}, we compare proof-nets with our diagrammatic calculus.\\
\linebreak
\bf{Acknowledgement:} I thank the anonymous referees for their many valuable comments.

%% file: contextualization.tex
\section{Contextualization}\label{context}
The present paper is a contribution to the investigation of the diagrammatic methods in the framework of the formal logical reasoning, see~\cite{275772}. The investigation of logical reasoning with diagrams has been pursued since a very long time, see~\cite{Gardner}, but commonly we tend to look at that form of reasoning as destined to a heuristic usage only, unsuitable for rigorous and formal deductive logical reasoning although well aware of the great immediacy with which the diagrammatic representations convey the informational contents. The primitive syntactic constituents of a typical logical formal system are linguistic symbols, but often various forms of diagrammatic representation support our deductive activity very usefully. Thus, any attempt of making a diagrammatic deductive activity mathematically rigorous has to be considered as naturally justified. 

The scientific community of mathematicians, and logicians in particular, became aware of the possibility of dealing rigorously with diagrammatic representations towards deductive logical reasoning near the end of last century, see~\cite{hammer1995logic, MR1312613, hammer:evl, MR1271699}, by introducing formal logical systems with diagrammatic primitive syntactic constituents. 

To contextualize this paper in the framework of logical reasoning with diagrams, in this section we recall some of the main features of the diagrammatic formal reasoning systems that extend the expressiveness of Euler diagrams, see~\cite{stapleton:a, Dau:2009:AFD:1560325.1560329}. We go from Euler diagrams to existential graphs, through Venn and Venn-Peirce diagrams, Venn-I and Venn-II diagrams, Euler/Venn diagrams, spider diagrams and constraint diagrams. 

Euler diagrams were introduced to illustrate the syllogistic reasoning, see~\cite{euler1843lettres}. Traditional syllogistic is based on the \emph{categorical propositions}
\[\begin{tabular}{lllll}
Each $A$ is $B$  &&&
Each $A$ is not $B$  \\
\\
Some $A$ is $B$ &&&
Some $A$ is not $B$
\end{tabular}\]
in which the letters $A, B$ are variables for meaningful expressions of the natural language referred to as \emph{terms}. In each categorical proposition, $A$ is 
the \emph{subject} and $B$ is the \emph{predicate}.
The categorical propositions ``Each $A$ is $B$'' and ``Each $A$ is not $B$'' are \emph{universal}, whereas the remaining two are \emph{particular} or \emph{existential}.

In Euler's system, the interior of a circle represents the extension of a term
and the logical relations among the terms in the categorical propositions are represented as inclusion, exclusion or intersection of circles, correspondingly as follows:
\[\begin{tikzpicture}
\draw [scale=.8] (.5,.5) circle (.7) (.5,1.03) node [text=black, scale=.8]{$B$};
\draw [scale =.8] (.5,.5) circle (.35) node (.5,.5) [text=black, scale=.8]{$A$};
\end{tikzpicture}
\qquad
\qquad
\begin{tikzpicture}
\draw [scale=.8] (.9,.5) circle (.7) (.9,.5) node [text=black, scale=.8]{$A$};
\draw [scale=.8] (2.5,.5) circle (.7) (2.5,.5) node [text=black, scale=.8]{$B$};
\end{tikzpicture}\]
\[\begin{tikzpicture}
\draw [scale=.8] (.5,.5) circle (.7) (.9,.5) node [text=black,scale=.8] {$A$};
\draw [scale=.8] (1.3,.5) circle (.7) (1.5,.5) node [text=black, scale=.8]{$B$};
\end{tikzpicture}
\qquad
\qquad
\begin{tikzpicture}
\draw [scale=.8] (.5,.5) circle (.7) (.3,.5) node [text=black, scale=.8]{$A$};  
\draw [scale=.8] (1.3,.5) circle (.7) (1.5,.5) node [text=black, scale=.8]{$B$};
\end{tikzpicture}\]
where it has to be observed how in contrast to the diagrams for the universal categorical propositions, those for the particular ones retain a certain amount of ambiguity. The correctness of the informational content conveyed by them heavily depends on where the letter ``$A$'' is written. As observed in~\cite{MR1312613}, the two diagrams become indistinguishable if the letter ``$A$'' is deleted.
Other drawbacks of Euler's system are discussed in~\cite{MR1312613}, among which the unclearity with which contradiction is expressed in it and the difficulty in unifying the informational content of two diagrams, such as the two diagrams for the premises of a syllogism. Unification of Euler diagrams is carefully formalized in~\cite{Mineshima2012-MINADI}, which we will comment on in section~\ref{motivation}.

A \emph{region}, or \emph{zone}, of an Euler diagram is any connected region of the plane inside some circle and outside the remaining circles, if present. Zones in Euler diagrams carry existential import, that is they are assumed to represent non-empty sets. This assumption seems to be the reason why a simple and unambiguous
correspondence between categorical propositions and Euler diagrams does not subsist.
Euler diagrams can express emptyness of sets through the employment of \emph{missing regions}, see~\cite{stapleton:a}, which are not diagrammatically represented since they have no area.
It is observed in~\cite{stapleton:a} that not all the approaches to Euler diagrammatic reasoning assume existential import in regions: in~\cite{hammer1995logic} one can find a sound and complete reasoning system based on Euler diagrams without existential import in regions.

In Venn diagrams, see~\cite{Venn1880}, the existential import in regions is removed and their emptyness is indicated by shading them. Shading is the syntactic device introduced for the purpose, so that Venn diagrams extend a fragment of the language of Euler diagrams. For instance, the Venn diagram for ``Each $A$ is $B$'' is 
\[\begin{tikzpicture}[fill=black!30, scale=.8]
\scope
\clip (-1,-.5) rectangle (2.8,1.6)
      (1.3,.5) circle (.7);
\fill (.5,.5) circle (.7);
\endscope
\draw (.5,.5) circle (.7) (.5,1.2)  node [text=black,scale=.8,above] {$A$}
      (1.3,.5) circle (.7) (1.4,1.2)  node [text=black,scale=.8,above] {$B$}
      (-1,-.5) rectangle (2.8,1.6);
\end{tikzpicture}\]
which really means ``There is nothing which is $A$ and not $B$''. Moreover, observe the introduction of the syntactic device which is the rectangle, to represent a background set.
Venn diagrams cannot represent existential import nor disjunctive information. Peirce modified them to allow such possibilities by the introduction of the syntactic devices of $x$-sequences to represent existential import and of $o$-sequences to represent the absence of elements. From left to right, the Venn-Peirce diagrams for the categorical propositions ``Some $A$ is $B$'' and ``No $A$ is $B$'' follow:
\[\begin{tikzpicture}[scale=.8]
\draw (.5,.5) circle (.7) (.5,1.2)  node [text=black,scale=.8,above] {$A$}
      (1.3,.5) circle (.7) (1.4,1.2)  node [text=black,scale=.8,above] {$B$}
      (-1,-.5) rectangle (2.8,1.6);
\draw (.9,.5) node [text=black]{$x$};
\end{tikzpicture}
\qquad
\begin{tikzpicture}[scale=.8]
\draw (.5,.5) circle (.7) (.5,1.2)  node [text=black,scale=.8,above] {$A$}
      (1.3,.5) circle (.7) (1.4,1.2)  node [text=black,scale=.8,above] {$B$}
      (-1,-.5) rectangle (2.8,1.6);
\draw (.9,.5) node [text=black]{$o$};
\end{tikzpicture}\]
Disjunctive information is represented by straight lines connecting $x$'s and $o$'s, see~\cite{stapleton:a}.

In~\cite{MR1312613} the system Venn-I is introduced. It is based on a modification of the Venn-Peirce diagrams that returns to shading to represent empty regions
and uses $\otimes$-sequences to indicate existential import in them. See also~\cite{MR1408440},~\cite{MR1271699} and~\cite{MR1408441}. The Venn-I diagrams for the four categorical propositions are 
\[\begin{tikzpicture}[fill=black!30, scale=.8]
\scope
\clip (-1,-.5) rectangle (2.8,1.6)
      (1.3,.5) circle (.7);
\fill (.5,.5) circle (.7);
\endscope
\draw (.5,.5) circle (.7) (.5,1.2)  node [text=black,scale=.8,above] {$A$}
      (1.3,.5) circle (.7) (1.4,1.2)  node [text=black,scale=.8,above] {$B$}
      (-1,-.5) rectangle (2.8,1.6);
\end{tikzpicture}
\qquad
\begin{tikzpicture}[fill=black!30, scale=.8]
\scope
\clip 
      (.5,.5) circle (.7);
\fill (1.3,.5) circle (.7);
\endscope
\draw (.5,.5) circle (.7) (.5,1.2)  node [text=black,scale=.8,above] {$A$}
      (1.3,.5) circle (.7) (1.4,1.2)  node [text=black,scale=.8,above] {$B$}
      (-1,-.5) rectangle (2.8,1.6);
\end{tikzpicture}\]
\[\begin{tikzpicture}[scale=.8]
\draw (.5,.5) circle (.7) (.5,1.2)  node [text=black,scale=.8,above] {$A$}
      (1.3,.5) circle (.7) (1.4,1.2)  node [text=black,scale=.8,above] {$B$}
      (-1,-.5) rectangle (2.8,1.6);
\draw (.9,.5) node {$\otimes$};
\end{tikzpicture}
\qquad
\begin{tikzpicture}[scale=.8]
\draw (.5,.5) circle (.7) (.5,1.2)  node [text=black,scale=.8,above] {$A$}
      (1.3,.5) circle (.7) (1.4,1.2)  node [text=black,scale=.8,above] {$B$}
      (-1,-.5) rectangle (2.8,1.6);
\draw (.3,.5) node {$\otimes$};
\end{tikzpicture}\]
correspondingly. 
An example of a formal diagrammatic proof of a syllogism in Venn-I together with the explicit indication of the rules of inference employed and the omission of rectangles, follows:
\def\tta[#1,#2]#3#4#5{\begin{scope}[scale=.8,xshift=#1,yshift=#2]
#5;
\draw (0,0) circle (20) (0,20) node [scale=.8,above] {$#3$}
(20,0) circle (20) (25,20) node [scale=.8,above] {$#4$};
\end{scope}}
\def\tA[#1,#2]#3#4;{\tta[#1,#2]{#3}{#4}{\begin{scope}
\clip (-30,-30) rectangle (50,45)
(20,0) circle (20);
\fill (0,0) circle (20);
\end{scope}}}
\def\tI[#1,#2]#3#4;{\tta[#1,#2]{#3}{#4}{\node at (10,0) {$\otimes$}}}
\def\tb(#1)#2#3{\begin{scope}[scale=.8,xshift=#1,yshift=0]
#2
\begin{scope}
\clip (-30,-30) rectangle (50,45)
(20,20) circle (20);
\fill (10,0) circle (20);
\end{scope}
\draw (0,20) circle (20) (0,40) node [scale=.8,above] {$S$}
 (20,20) circle (20) (20,40) node [scale=.8,above] {$P$}
 (10,0) circle (20) (10, -20) node [scale=.8,below] {$M$};
#3
\end{scope}}
\def\mt(#1)#2#3{\node at (#1,20) {#2};
\node at (#1,12) {#3};
\node at (#1,0) {$\vdash$}}
\def\mo(#1)#2{\node at (#1,12) {#2};
\node at (#1,0) {$\vdash$}}
\begin{center}
\begin{tikzpicture}[x=1pt,y=1pt,fill=black!30, scale=.8]
\tA[-2,30]MP;
\tI[-2,-30]SM;
\draw (35,10) -- (60,2) ;
\draw (35,-10) -- (60,-2) ;
\mo(66){\footnotesize unif.};
\tb(119){}{\node at (-3,6) {$\otimes$} ;
\node at (10,12) {$\otimes$} ;
\draw (-.5,7) -- (7,11) ;}
\mt(145){\footnotesize erasure}{\footnotesize of link};
\tb(222){}{\node at (10,12) {$\otimes$} ;}
\mo(232){\footnotesize extension};
\end{tikzpicture}
\begin{tikzpicture}[x=1pt,y=1pt,fill=black!30, scale=.8]
\mo(2){\footnotesize extension};
\tb(50){}{\node at (10,28) {$\otimes$} ;
\node at (10,12) {$\otimes$} ;
\draw (10,15) -- (10,25) ;}
\mt(99){\footnotesize erasure of}{\footnotesize closed curve};
\tI[175,0]SP;
\end{tikzpicture}
\end{center}
The unification rule captures the conjuction of the semantic contents in the diagrams it applies. The erasure of link rule allows the erasure of a part of 
an $\otimes$-sequence if that part is in a shaded zone. The erasure of closed curve is an instance of the rule of erasure of a diagrammatic object which is a shading, an $\otimes$-sequence or a closed curve. The extension rule allows the introduction of an extra link to an existing $\otimes$-sequence. To handle contradiction, among the inference rules of Venn-I there is also the \emph{rule of conflicting information}: if a diagram has a zone with both a shading and an entire $\otimes$-sequence, then any diagram follows. 

In order to improve its expressiveness, the system Venn-I has been further extended to the system Venn-II for the diagrammatic representation of the informational content in statements with truth-functional connectives such as ``or'' and ``if\ldots, then\ldots''. In Venn-II, Venn-I diagrams are allowed to be connected by straight lines, see~\cite{MR1312613}.

Euler/Venn diagrams, see~\cite{swoboda:ievrs, Swoboda:2005:HRE:1705531.1705841}, represent a modification of Venn-I diagrams which consists of diagrams with both Euler-like and Venn-like features. Shadings and $\otimes$-sequences are available but named constant sequences are available as well to represent the presence of particular individuals in a set. For instance, the Euler/Venn diagram
\[\begin{tikzpicture}[scale=.8]
\scope
\clip 
      (.5,.5) circle (.7);
\fill [fill=black!30] (1.3,.5) circle (.7);
\endscope
\draw (.5,.5) circle (.7) (0,1.2)  node [text=black,scale=.8,above] {$Mammals$}
      (1.3,.5) circle (.7) (1.9,1.2)  node [text=black,scale=.8,above] {$Insects$}
      (-1,-.5) rectangle (2.8,1.6);
\draw (.2,.5) node [text=black, scale=.8] {$tim$};
\draw (1.6, .5) node [text=black, scale=.8] {$tim$};
\draw (0.44,0.45)--(1.37,0.45);
\end{tikzpicture}\]
which we took from~\cite{stapleton:a}, indicates that no individual is both a mammal and an insect and that there is an individual named ``tim'' which is either a mammal or an insect. Named constant sequences permit a much better correspondence between the diagrams and the information carried by formulas in first order logic. In~\cite{swoboda:udttvevhaevfhroi:02}, an algorithm is given for verifying if an Euler/Venn monadic first order formula is observable from an Euler/Venn diagram. Observable formulas are those which are consequence of the information displayed in the diagram, see~\cite{swoboda:ticoevd}. See also~\cite{swoboda:acsotdaiohrs, swoboda:mhs}.

Spider diagrams, see~\cite{molina2001reasoning, Howse01spiderdiagrams:, Howse99reasoningwith, Howse00onthe, Gil1999a, Howse:2005:SD}, adapt and extend the Venn-II diagrams. In place of $\otimes$-sequences to indicate non-empty regions of diagrams, \emph{spiders} are used. Distinct spiders represent distinct elements, so that they provide finite lower bounds on the cardinalities of the represented sets. Finite upper bounds on the cardinalities of the represented sets are given by shading together with spiders: spiders in a shaded region represent all the elements in the set represented by that region. For instance, we took the leftmost of the spider diagrams
\[\begin{tikzpicture}[scale=.8]
\scope
\clip 
      (.5,.5) circle (.7);
\fill [fill=black!30] (1.3,.5) circle (.7);
\endscope
\draw (.5,.5) circle (.7) (0,1.2)  node [text=black,scale=.8,above] {$Cars$}
      (1.3,.5) circle (.7) (1.9,1.2)  node [text=black,scale=.8,above] {$Vans$}
      (-1,-.5) rectangle (2.8,1.6);
\draw (.1,.5) node [scale=.8]{$\b$};
\draw (1.7,.5) node [scale=.8]{$\b$};
\draw (.1,.3) node [scale=.8]{$\b$};
\draw (0.1,0.5)--(1.7,0.5);
\end{tikzpicture}
\qquad
\begin{tikzpicture}[scale=.8]
\scope
\clip (-1,-.5) rectangle (2.8,1.6)
      (1.3,.5) circle (.7);
\fill [fill=black!30] (.5,.5) circle (.7);
\endscope
\draw (.5,.5) circle (.7) (0,1.2)  node [text=black,scale=.8,above] {$Cars$}
      (1.3,.5) circle (.7) (1.9,1.2)  node [text=black,scale=.8,above] {$Vans$}
      (-1,-.5) rectangle (2.8,1.6);
\draw (.1,.5) node [scale=.8]{$\b$};
\draw (.1,.3) node [scale=.8]{$\b$};
\end{tikzpicture}\]
from~\cite{Stapleton04whatcan}. It expresses that no element is both a car and a van and that there are at least two elements one of which is a car and the other is a car or a van. The rightmost spider diagram expresses that there are exactly two cars which are not vans. The possibility of expressing lower and upper bounds on the cardinalities of sets makes the spider diagrams more expressive than the Venn-II diagrams. In~\cite{Stapleton04whatcan, Stapleton:2004:ESD:1094471.1094501} it is shown that a suitable fragment of monadic first order logic with equality and no function or constant symbols is equivalent, in expressive power, to the language of spider diagrams. Several spider diagrams systems exist, see~\cite{stapleton:a, molina2001reasoning, howse:sasacdrs}. In~\cite{DBLP:conf/iccs/DauF08} spider diagrams are combined with conceptual graphs, see~\cite{Dau2009a, Sowa:1984:CSI:4569}, giving rise to a hybrid diagrammatic notation called \emph{conceptual spider diagrams}.

Constraint diagrams, see~\cite{Kent_97_Constraint, GilHK99, GilHK01, Stapleton03aconstraint, Stapleton04reasoningwith}, were proposed as a notation for specifying constraints in object oriented software modelling. They extend spider diagrams by allowing the representation of  universal quantification through the employment of universal spiders indicated by asterisks, arrows representing properties of binary relations and derived contours that are not labelled, must be the target of an arrow and represent the image of a relation when the domain is restricted to the source, which can be a spider, a contour or a derived contour. For instance, the constraint diagram
\[\begin{tikzpicture}[scale=.8]
\draw (-3,-.5) rectangle (2.8,1.6);
\draw (-2,.7) circle (.5) (-2,1.36) node [text=black,scale=.8]{$Members$};
\draw (1.7,.5) circle (.7) (1.7, 1.36) node [text=black,scale=.8]{$Videos$};
\draw (1.7,.5) circle (.5);
\draw (-1.8,.6) node [scale=.8]{$\ast$};
\draw (1.7,.5) node [scale=.8]{$\b$};
\draw [->, >=latex] (-1.8,.6) ..controls (-0.3,0)..(1.2,.5); 
\draw (-0.3,0) node [text=black,scale=.8]{$past~rentals$};
\end{tikzpicture}\]
which we took from~\cite{stapleton:a}, contains a universal spider which is the source of an arrow representing a relation whose image is a derived countour. The diagram specifies that no members are videos and for each member, all the past rentals, of which there is at least one, are videos.

Conflicting information in the systems of Euler/Venn diagrams, spider diagrams and constraint diagrams is dealt with by letting any diagram follow from any pair of inconsistent diagrams.

Existential graphs, see~\cite{Dau06mathematicallogic, Dau:2009:AFD:1560325.1560329, roberts1973existential, zeman1964graphical, shin2002iconic} were introduced by Peirce as a diagrammatic system for representing and reasoning with relations. The system is divided into three parts named Alpha, Beta and Gamma. Alpha and Beta
respectively correspond to propositional logic and first order logic with predicates and equality without functions or constants, whereas Gamma includes features of higher order logic and modal logic, as well as the possibility of expressing self-reference. Gamma was never finished by Peirce and, as pointed out in~\cite{Dau:2009:AFD:1560325.1560329}, mainly its modal logic fragment has been elaborated to the contemporary mathematical standards. We here focus our interest on Alpha and Beta, which are proved to be sound and complete in~\cite{Dau06mathematicallogic}.

The existential graphs of Alpha consist of the following two syntactic devices: atomic propositions and cuts. The latter are represented by fine-drawn, closed and doublepoint-free curves. Different graphs on the plane express their conjunction whereas the inclusion of a graph in a cut denotes its negation. For instance, we took from~\cite{Dau:2009:AFD:1560325.1560329} the following two existential graphs of Alpha:
\[\begin{tikzpicture}[scale=.8]
\draw [opacity = 0] (0,0) rectangle (2,2);
\draw (1,2) node [scale=.8] {it is stormy};
\draw (-1,1) node [scale=.8] {it rains};
\draw (-1,2) node [scale=.8] {it is cold};
\end{tikzpicture}
\qquad
\qquad
\begin{tikzpicture}[scale=.8]
\draw [rounded corners=2ex] (0,.7) rectangle (5,2.5);
\draw (4,2) node [scale=.8] {it is stormy};
\draw (2,1) node [scale=.8] {it rains};
\draw (1,2) node [scale=.8] {it is cold};
\draw [rounded corners=1ex] (.2,1.7) rectangle (2,2.2);
\end{tikzpicture}\]
in the leftmost, the three propositions ``it is stormy'', ``it rains'' and ``it is cold'' have been written near each other so that in the graph they express the single proposition which is their conjunction, that is `` it is stormy and it rains and it is cold''. The rightmost diagram expresses the fact that it is not true that it rains and it is stormy and it is not cold, namely that ``if it rains and it is stormy, then it is cold''. 

In Beta, predicates of arbitrary arity can be used and the syntactic device which is the \emph{line of identity} is introduced. Lines of identities are represented as heavily drawn lines and indicate both the existence of objects and the identity between them. For instance, the graph
\[\begin{tikzpicture}[scale=.8]
\draw [very thick] (0,0)--(1,0) (-.2,0) node [scale=.8] {cat}  (1.2,0) node [scale=.8] {on};
\draw [very thick] (1.38,0)--(1.94,0) (2.2,0) node [scale=.8] {mat};
\end{tikzpicture}\]
which we took from~\cite{Dau06mathematicallogic}, contains two lines of identity. It denotes two not necessarily different objects. The first object is the unary predicate ``cat'', so it denotes a cat, whereas the second object denotes a mat. Both the objects are linked to the binary predicate ``on''. Therefore, the first object is in the relation ``on'' to the second. The graph means that ``there are a cat and a mat such that the cat is on the mat'', namely ``a cat is on a mat''. Lines of identity can be combined to networks forming \emph{ligatures}. For instance, in the leftmost of the graphs
\[\begin{tikzpicture}[scale=.8]
\draw [very thick] (0,-.15)--(0,-.91) (0,0) node [scale=.8]{cat} (0,-1) node [scale=.8]{cute};
\draw [very thick] (0,-0.5)--(.6,-.5) (1,-.5) node [scale=.8]{young};
\end{tikzpicture}
\qquad
\qquad
\begin{tikzpicture}
\draw [very thick] (0,0)--(1.3, 0) (-.3,0) node[scale=.8]{man} (1.85, 0) node [scale=.8]{will die};
\draw [rounded corners=1ex] (1.1, -.3) rectangle (2.5,.3);
\end{tikzpicture}\]
there is a ligature composed of three lines of identity which denotes one object. Its meaning is ``there is a cat which is cute and young'', whereas the meaning of the rightmost graph above is ``there exists a man who will not die''. The two graphs were taken from~\cite{Dau06mathematicallogic} and~\cite{Dau:2009:AFD:1560325.1560329}, respectively.

%% file: motivation.tex
\section{Motivation}\label{motivation}
This paper deals with a formal diagrammatic treatment of a reading of the syllogistics in a fragment of the multiplicative intuitionistic linear logic. 
To motivate this paper, in this section we compare our diagrammatic approach to syllogistics with the one developed in some recent contributions in the field and discuss its connections with linear logic.

In~\cite{Mineshima:2008:DRS:1432522.1432547} and~\cite{Mineshima2012-MINADI} the Generalized Diagrammatic Syllogistic inference system GDS is extensively investigated. GDS has two inference rules which are \emph{unification} and \emph{deletion} and is syntactically based on the diagrammatic representation system EUL. The latter is a conceptually new formalization of Euler diagrammatic reasoning  in which diagrams are defined in terms of topological relations between diagrammatic objects, rather than in terms of regions as in the case of Venn, Euler/Venn, spider and constraint diagrams, see section~\ref{context}.
In~\cite{Mineshima:2008:DRS:1432522.1432547} the cognitive aspects of the reasoning with EUL diagrams with emphasis on the deductive syllogistic activity are investigated besides the logical ones. See also~\cite{SatoMT10, SatoCognEff}.

The diagrammatic representation systems can be distinguished with respect to the possibility of expressing negation, existence or disjunction in them, through the employment of suitable syntactic devices or not. We recall from section~\ref{context} that Venn diagrams can express negation by shading regions, whereas they cannot express existence or disjunction. Euler diagrams cannot express disjunction, they use exclusion of circles to express negation, they express existence by assuming existential import in their regions.

The system EUL is based on Euler's circles. It uses topological relations between them to express negative information and ``named points'' to indicate the existence of a particular element. It cannot express disjunctive information.
The EUL diagrams for the four categorical propositions are 
\[\begin{tikzpicture}
\draw [scale=.8] (-.3,-.3) rectangle (1.3,1.3);
\draw [scale=.8] (.5,.5) circle (.7) (.5,1.03) node [text=black, scale=.8]{$B$};
\draw [scale =.8] (.5,.5) circle (.35) node (.5,.5) [text=black, scale=.8]{$A$};
\end{tikzpicture}
\qquad
\begin{tikzpicture}
\draw [scale=.8] (.1,-.3) rectangle (3.3,1.3);
\draw [scale=.8] (.9,.5) circle (.7) (.9,.5) node [text=black, scale=.8]{$A$};
\draw [scale=.8] (2.5,.5) circle (.7) (2.5,.5) node [text=black, scale=.8]{$B$};
\end{tikzpicture}\]
\[\begin{tikzpicture}
\draw [scale=.8] (-.3,-.3) rectangle (2.1,1.3);
\draw [scale=.8] (.5,.5) circle (.7) (.3,.5) node [text=black,scale=.8] {$A$};
\draw [scale=.8] (1.3,.5) circle (.7) (1.5,.5) node [text=black, scale=.8]{$B$};
\draw (.7,.5) node [scale=.8] {$x$};
\end{tikzpicture}
\qquad
\begin{tikzpicture}
\draw [scale=.8] (-.3,-.3) rectangle (2.1,1.3);
\draw [scale=.8] (.5,.5) circle (.7) (.3,.5) node [text=black, scale=.8]{$A$};  
\draw [scale=.8] (1.3,.5) circle (.7) (1.5,.5) node [text=black, scale=.8]{$B$};
\draw (.3,.1) node [scale=.8] {$x$};
\end{tikzpicture}\]
correspondingly. A set-theoretic semantics for EUL diagrams is introduced to formally capture the idea that EUL circles correspond to predicates and named points to constant symbols. As pointed out in~\cite{Mineshima:2008:DRS:1432522.1432547}, contrary to the traditional treatment of syllogistics in the natural language, the EUL diagrams for the categorical propositions do not retain the subject-predicate distinction which is in them. Namely, the EUL diagram for ``Some $A$ is $B$'' represents ``Some $B$ is $A$'' as well, and similarly for ``No $A$ is $B$''. 
 
The unification of diagrams and the extraction of information from them play a central r\^ole in the construction of diagrammatic proofs, of syllogisms in particular. In EUL-style diagrammatic proofs the unification of diagrams introduces disjunctive ambiguities in the placement of named points in regions whereas in Venn-style diagrammatic proofs it does not, since Venn diagrams can be unified essentially by simply superposing shaded regions. On the other hand, EUL diagrams retain the clarity of the representations and inferences of Euler diagrams whereas Venn diagrams easily do not. See~\cite{hammer:evl}.

In~\cite{Mineshima:2008:DRS:1432522.1432547, Mineshima2012-MINADI} a unification rule for the EUL-based diagrammatic inference system GDS which does not generates disjunctive ambiguities is suitably formalized. The unification operation is decomposed into more primitive unification steps and the ensuing schemes of the rule of unification are carefully distinguished on the base of the number and type of the diagrammatic objects shared by the diagrams the rule is applying, and
divided into three groups. The onset of possible disjunctive ambiguities is avoided through the imposition of a \emph{constraint for determinacy}: the unification of any two diagrams is not permitted when the relations between each point and all the circles of the two diagrams are not determined. 
A further \emph{constraint for consistency} is imposed on unification: it is not permitted to unify two inconsistent diagrams.

The diagrammatic proofs in GDS consist essentially of unification and deletion steps only. Suitable restrictions of the unification and deletion rules involving exclusively the previous EUL diagrams permit the identification of a Diagrammatic Syllogistic inference sub-system DS of GDS. Syllogistic proofs in DS can be put in a normal form that consists of alternating and not repeating unification and deletion steps. In~\cite{Mineshima:2008:DRS:1432522.1432547} it is proved that a syllogism is valid if and only if it is provable in DS, whereas in~\cite{Mineshima2012-MINADI} it is proved that GDS is sound and complete. For instance, here is the proof in DS of the syllogism whose proof in Venn-I has been given in section~\ref{context}:
\[\begin{tikzpicture}[scale=.8]
\draw [scale=.8] (-.3,-.3) rectangle (1.3,1.3);
\draw [scale=.8] (.5,.5) circle (.7) (.5,1.01) node [text=black, scale=.8]{$P$};
\draw [scale =.8] (.5,.5) circle (.3) node (.5,.5) [text=black, scale=.8]{$M$};

\draw [scale=.8] (3.3,-.3) rectangle (5.7,1.3);
\draw [scale=.8] (4.1,.5) circle (.7) (3.9,.5) node [text=black,scale=.8] {$S$};
\draw [scale=.8] (4.9,.5) circle (.7) (5.1,.5) node [text=black, scale=.8]{$M$};
\draw (3.6,.5) node [scale=.8] {$x$};

\draw [->, >=latex, scale=.8] (.5,-.4) -- (1.6,-1);
\draw [->, >=latex, scale=.8] (3.6,-.4) -- (1.8,-1);

\draw (1.5,-.4) node [text=black, scale=.8]{unif.};
\draw [scale=.8] (.3,-3) rectangle (2.9,-1.1);
\draw [scale=.8] (1.1,-2) circle (.7) (1,-2) node [text=black,scale=.8] {$S$};

\draw (1.32,-1.6) node [scale=.8, text=black]{$x$}; 

\draw [scale=.8] (1.88,-2) circle (.4) (2,-2) node [text=black,scale=.8]{$M$};
\draw [scale=.8] (2.1,-2) circle (.7) (2.5,-2) node [text=black,scale=.8] {$P$};

\draw [scale=.8] (.3,-5.5) rectangle (2.9,-3.8);
\draw [->, >=latex, scale=.8] (1.32,-3.1) -- (1.32, -3.7);
\draw (1.8, -2.7) node [text=black,scale=.8] {deletion};
\draw [scale=.8] (1.1,-4.7) circle (.7) (.9,-4.7) node [text=black,scale=.8] {$S$};
\draw [scale=.8] (2.1,-4.7) circle (.7) (2.3,-4.7) node [text=black, scale=.8]{$P$};
\draw (1.3,-3.7) node [scale=.8] {$x$};

\end{tikzpicture}\]
Finally it has to be observed that the diagrammatic reasoning in EUL does not have explicit operations with negation and it does not have explicit operations with  subject-predicate distinction.

In~\cite{Pagnan2012-PAGADC} the diagrammatic logical system SYLL
for the traditional syllogistics is introduced. Its syntactic primitives consist of the symbols $\b$, $\arr$, $\leftarrow$, as well as of countably many term-variables $A,B,C,\dots$. A SYLL diagram is a finite list of arrow symbols separated by a single bullet symbol or term-variable, beginning and ending at a term-variable. Thus SYLL diagrams are graphically linear. For instance, the following is a SYLL diagram: $A\arr\b\arr\b\arr B\arr\b\leftarrow C$.

A \emph{part} of a SYLL diagram is a finite list of consecutive components of it. For instance, the following are parts of the previous SYLL diagram: $\arr\b\arr$, $A\arr\b$, $\b\arr B\arr\b$.

SYLL diagrams are not defined in terms of topological relations nor in terms of regions. Nonetheless, they can express negative information, existential import and disjunctive information. 

The \emph{complement} of a term is the term itself preceded by a ``non'' adverb of modality, so that for instance the complement of ``man'' is ``non-man''. A part of a SYLL diagram such as $A\arr\b$ or $\b\leftarrow A$ represents the complement of the term-variable $A$, which is ``non-$A$''.

The SYLL diagrams for the four categorical propositions are the following:
\[\begin{tabular}{lllll}
& \AA{A}{B} &&
& \EE{A}{B}\\
\\
& \II{A}{B} &&
& \OO{A}{B}
\end{tabular}\]
correspondingly. The complement of the term-variable $B$ occurs in the SYLL diagrams for ``Each $A$ is not $B$'' and ``Some $A$ is not $B$'', since $B$ fits in the part $\b\leftarrow B$ of both. Because of this, those categorical propositions should really be read affirmatively as ``Each $A$ is non-$B$'' and ``Some $A$ is non-$B$'', respectively. 

This is the starting point for the extension of SYLL to the diagrammatic reasoning system for the modern nineteenth century syllogistics with complemented terms which we will treat in section~\ref{demorgan}. We cannot see how to extend
the reasoning systems which are based on Euler diagrams 
to syllogistics with complemented terms, as well as more in general how to 
let negation be explicitly represented in them, but see~\cite{Stapleton_incorporatingnegation}.

The SYLL diagrams for the categorical propositions do not retain the subject-predicate distinction in them. This is due to the symmetric look of the SYLL diagrams 
\[\II{A}{B}\qquad\EE{A}{B}\]
as it was the case for the corresponding EUL diagrams.

The so-called ``strengthened'' syllogisms are those that proceed from two universal premises towards an existential conclusion. They are valid if and only if existential import on one of the terms involved is assumed.

Individuals are not foreseen to be representable in SYLL, in accordance with a conception of syllogistic more adherent to the original one, according to which singular names should be avoided. Thus, existential import is not explicitly expressed through the placement of a named point inside the interior of a circle representing the extension of a term but implicitly, referring intensionally to such an extension through the employment of a special SYLL diagram. That is, existential import is carried by a term-variable in virtue of its occurrence in a special SYLL diagram. 
Similar considerations hold for the expressibility of emptyness
and contradiction.
The SYLL diagram that expresses existential import on a term-variable $A$ is
\[\II{A}{A}\]
whereas the one that expresses the ``emptyness'' of $A$ is
\[\EE{A}{A}\]
The SYLL diagram that expresses contradictory informational content on $A$ is
\[\OO{A}{A}\]

Disjunctive information linking two term-variables $A, B$ is expressible in the diagrammatic formalism of SYLL in the form ``Each non-$A$ is $B$'' by the SYLL diagram $A\arr\b\arr B$, which does not equivalently express ``Each non-$B$ is $A$''. It will be put in connection with a non-symmetric multiplicative disjunction in linear logic. See remark~\ref{gh}.

The diagrammatic operation in SYLL that by analogy corresponds with unification in GDS is \emph{concatenation} of SYLL diagrams, not by juxtaposition but by superposition of their extremal term-variables. These serve as actual terms in the sense of the original meaning of the latin word \emph{terminum}, which is boundary. Any two SYLL diagrams with one coincident extremal term-variable can be concatenated without the imposition of any kind of constraint. The SYLL diagram which is the concatenation of two or more SYLL diagrams captures the conjunction of the informational content in each of them without the onset of any sort of ambiguity. For instance, the SYLL diagram 
\[A\arr\b\leftarrow B\arr C\leftarrow\b\leftarrow D\]
is the concatenation of the SYLL diagrams $A\arr\b\leftarrow B$, $B\arr C$, $C\leftarrow\b\leftarrow D$.

Essentially, proofs in SYLL consist of concatenation and deletion steps. In particular, for a syllogism: the SYLL diagrams for its premises are concatenated by superposing the term-variable representing the middle term which, if possible, is deleted towards the obtainment of the SYLL diagram for the conclusion. The deletion of a middle term-variable $M$ in a proof in SYLL is allowed if and only if it occurs in a part such as $\arr M\arr$ or $\leftarrow M\leftarrow$ of the SYLL diagram which is the concatenation of the SYLL diagrams for the premises, by substituting that part with just one accordingly oriented arrow symbol $\arr$ or $\leftarrow$, respectively. For instance, here is the proof in SYLL of the syllogism whose proof in DS has been shown above:
\[\AxiomC{$S\leftarrow\b\arr M$}
\AxiomC{$M\arr P$}
\BinaryInfC{$S\leftarrow\b\arr M\arr P$}
\UnaryInfC{$S\leftarrow\b\arr P$}
\DisplayProof\]
The implementation of a proof in SYLL is algorithmic in a naive sense. Proofs in SYLL are carried out in accordance with the rules of inference displayed in definition~\ref{rules}. The conjunction of the informational content of two SYLL diagrams through the operation of concatenation does not require any articulated spatial rearrangemet of the diagrammatic objects involved. Subsequently, the deletion of a middle term-variable is subject to a criterion which is just ``look at the orientation of the arrow symbols surrounding the middle term-variable''. Nothing peculiar of an intelligent being enters in the execution of concatenation and deletion.

In~\cite{Pagnan2012-PAGADC} we proved that a (strengthened) syllogism is valid if and only if it is provable in SYLL. We were able to achieve this result also by pointing out a diagrammatic criterion for the rejection of the invalid schemes of syllogism. Such a criterion is based on the following observation: in the implementation of a SYLL proof as many bullet symbols are in the SYLL diagram for its conclusion as in the SYLL diagrams for its premises. 
For instance, on the base of this criterion we are able to immediately establish that no existential conclusion can follow from the premises $S\arr M$, $M\arr P$, because the SYLL diagram for any one such conclusion contains at least one bullet symbol but, by adding existential import on the term-variable $S$ as a further hypothesis, namely through the SYLL diagram $S\leftarrow\b\arr S$, we have 
\[\AxiomC{}
\UnaryInfC{$S\leftarrow\b\arr S$}
\AxiomC{$S\arr M$}
\AxiomC{$M\arr P$}
\BinaryInfC{$S\arr M\arr P$}
\UnaryInfC{$S\arr P$}
\BinaryInfC{$S\leftarrow \b\arr S\arr P$}
\UnaryInfC{$S\leftarrow\b\arr P$}
\DisplayProof\]
which is a proof in SYLL of the valid strengthened syllogism $\aa MP,\aa SM, \ii SS\therefore\ii SP$ occurring in the lower half of the first column of the table~(\ref{questa}) in section~\ref{syllru}, whose medieval name is 
bArbArI.

Returning on the handling of inconsistency, we have observed that in GDS an imposed constraint for consistency does not permit the unification of two inconsistent EUL diagrams. In other systems a linguistic sign such as $\bot$ is introduced for the purpose, see~\cite{swoboda:udttvevhaevfhroi:02, MR1408441, Howse:2005:SD}. In SYLL, the inconsistency of two diagrams such as $A\leftarrow\b\arr B$, $A\arr\b\leftarrow B$, for instance, is provable:
\[\AxiomC{$A\leftarrow\b\arr B$}
\AxiomC{$A\arr\b\leftarrow B$}
\UnaryInfC{$B\arr\b\leftarrow A$} 
\BinaryInfC{$A\leftarrow\b\arr B\arr\b\leftarrow A$}
\UnaryInfC{$A\leftarrow\b\arr\b\leftarrow A$}
\DisplayProof\]

SYLL extends to syllogistic reasoning with an arbitrarily fixed positive integer number $n$ of term-variables, see~\cite{Meredith, MR0107598}. An $n$-term (strengthened) syllogism is valid if and only if it is provable in SYLL. This result seems to be peculiar of the diagrammatic inference systems based on linear diagrams, see~\cite{Smyth}. See also~\cite{MR1149957, MR1776228}.

In~\cite{Lukasiewicz}, an axiomatic systematization of Aristotle's syllogistic within propositional logic is carried out. In~\cite{MR0317886,MR0497848, Smiley1973-SMIWIA} it is made apparent that such an axiomatic approach is unsuitable because, closer to Aristotle's original conception, syllogistics is shown to be a sound and complete deductive system which seems to be independent of propositional or first order logic. Thus, the investigation of it in a well-known fragment of a logical system, as adherently as possible to his original conception, seems to be desirable. This is our case with respect to a fragment of linear logic that we will be dealing with in more detail in the subsequent sections. Linear logic gave new insight into basic notions that seemed to be established forever. The investigation of syllogistics within it could reveal itself unexpectedly fruitful.

Linear logic, see~\cite{MR899269,MR1356006}, 
is an extension of both classical and intuitionistic logic which is obtained by the introduction of new connectives. 
Roughly speaking, such connectives are divided into \emph{multiplicative} and \emph{additive}. In the present paper we will deal mainly with two multiplicative connectives: the \emph{multiplicative conjunction} $\otimes$ and the \emph{linear implication} $\multimap$. 

A distinctive feature of linear logic with respect to classical and intuitionistic logic is the following: whereas the latter deal with stable truths, the former deals with evolving truths. In more concrete terms, given $A$ and $A\Arr B$, in classical or intuitionistic logic one is able to infer $B$, but $A$ is still available. That is, the sequent $A, A\Arr B\vdash A\and B$ is provable.

Linear implication is conceived as a causal one and because of this the information which is at its antecedent is not iteratively reusable in general. It corresponds to the availability of a resource whose employment means its exhaustion. 
Concretely, in linear logic the sequent $A, A\multimap B\vdash B$ is provable, whereas $A,A\multimap B\vdash A\otimes B$ is not. Multiplicative conjunction is conceived as the availability of two actions both of which will be done.  The sequent $A\multimap B, A\multimap C\vdash A\otimes A\multimap B\otimes C$ is provable, whereas the sequent $A\multimap B, A\multimap C\vdash A\multimap B\otimes C$ is not. 

Towards our reading of syllogistics within intuitionistic linear logic, we assume the availability of a distinguished atomic formula $\bot$ and, for every formula $A$, write $A^{\bot}$ as an abbreviation for $A\multimap\bot$.

Our reading of syllogistics within linear logic is founded on the \emph{categorical formulas}:
\[\begin{tabular}{llllll}
& $A\multimap B$ &&&
& $A\multimap B^{\bot}$\\
\\
& $A\otimes B$ &&&
& $A\otimes B^{\bot}$
\end{tabular}\]
correspondingly. They translate the Aristotelian categorical propositions more faithfully than any translation of them in the usual propositional or first order logical calculus. Explicitly, $A\multimap B$ can be read as ``every object $A$, used once and consumed, produces an object $B$'', whereas $A\otimes B$ can be read as ``some object $A$ coexists with some object $B$''. Similarly, $A\multimap B^{\bot}$ can be read as ``every object $A$, used once and consumed, produces the absence of $B$'', whereas $A\otimes B^{\bot}$ can be read as ``some object $A$ coexists with the absence of some object $B$''. 

The previous considerations can be found in~\cite{Abrusci}, together with the previous formulas translating the categorical propositions but we independently found them. 
Theorem~\ref{perquesta} will mathematically justify the admissibility of the categorical formulas as suitable translations of the categorical propositions.

The SYLL diagrams for the categorical propositions represent the previous readings of the categorical formulas very closely.

In~\cite{Abrusci}, the Aristotelian syllogistics is investigated in the multiplicative fragment of classical linear logic where the fundamental linear logic connective $(~)^{\bot}$ which is \emph{linear negation} is available since the beginning. The atomic formulas are the literals $A$ or $A^{\bot}$. Linear negation is involutive, in the sense that $A^{\bot\bot}$ is identifiable with $A$, whereas it is not in our framework which is intuitionistic. Moreover, our framework extends to De Morgan's syllogistics, see section~\ref{demorgan}

In~\cite{Abrusci}, the traditional syllogistics is investigated from the point of view of the internal dualities generated by the interrelation of the contradictory categorical propositions. Linear negation allows the explicit handling of such dualities, and makes apparent their involvement in the graphical treatment via proof-nets of the reduction rules of the syllogistics, see section~\ref{redrules}. We will deal with proof-nets in section~\ref{comm}.

%% file: prelim.tex
\section{Preliminaries on syllogistics}\label{syllru}
We recall some basics of the syllogistics from the standpoint of
its traditional medieval systematization and its nineteenth century modern systematization, with reference in particular to syllogistic with complemented terms to be discussed more extensively in section~\ref{demorgan}. See~\cite{ADeMorgan, de1911formal, MR2464674}.

Syllogistics in its original form dates back to Aristotle, who formalized it as a logical system in the Prior Analytics. It is not the intention of the present paper to investigate the deep nature of the original Aristotelian conception of the syllogistics. More refined approaches that treat it formally in a natural deduction style exist: see~\cite{MR0317886} and~\cite{Smiley1973-SMIWIA}.

Aristotle's syllogistics is based on the following four \emph{categorical propositions}:
\[\begin{tabular}{lllll}
&Each $A$ is $B$ &&Each $A$ is not $B$\\
\\
\\
&Some $A$ is $B$ &&Some $A$ is not $B$
\end{tabular}\]
in which the upper case letters $A$ and $B$ are variables for \emph{terms}, that is meaningful expressions of the natural language. More precisely, in each categorical proposition the letter $A$ denotes the term which is its \emph{subject} whereas $B$ denotes the term which its \emph{predicate}.

The \emph{complement} of a term is the term itself preceded by a ``non'' adverb of modality. Thus, the complement of the term ``man'' is ``non-man'', for instance. 
\begin{notation}
The complement of a term symbolically represented by a variable $A$ will be symbolically represented by the corresponding lower case letter $a$ to mean non-$A$, in accordance with De Morgan's notation. 
\end{notation}
Upper case and lower case letters respectively denoting terms and their complements will be henceforth referred to as \emph{term-variables}.

Traditional syllogistic requires the subject of a categorical proposition to be 
a term which is not a complement, whereas it allows the predicate of a categorical proposition to be a term which is a complement. Thus, the modern reading of the previous traditional categorical propositions is the following:
\[\begin{tabular}{lllll}
&\aa AB: Each $A$ is $B$ && \aa Ab:  Each $A$ is non-$B$\\
\\
\\
&\ii AB: Some $A$ is $B$ && \ii Ab:  Some $A$ is non-$B$
\end{tabular}\]
all of which are now affirmative. In fact, the boldface letters \textbf{A}, \textbf{I} refer to the corresponding vowels in the latin word ``adfirmo''.
\begin{notation}
Arguments in the natural language will be written $P_1,\ldots,P_n\therefore C$, where the $P_i$'s are the premises and $C$ is the conclusion.
\end{notation}
A \emph{valid argument} is one whose conclusion necessarily follows from its premises, independently from their meaning. 
\begin{definition}
A \emph{syllogism} is an argument $P_1,P_2\therefore C$ where $P_1,P_2,C$ are categorical propositions that are distinguished in \emph{first premise}, \emph{second premise} and \emph{conclusion}, from left to right. 
Moreover, in them three term-variables among $S$, $M$, $P$, $s$, $m$, $p$ occur as follows: $M$ (resp. $m$) occurs in both the premises and does not occur in the conclusion whereas, according to the traditional way of writing syllogisms, $P$ (resp. $p$) occurs in the first premise and $S$ (resp. $s$) occurs in the second premise. The term-variables $S$ (resp. $s$) and $P$ (resp. $p$) occur as the subject and predicate of the conclusion, respectively, and are referred to as \emph{minor term} and \emph{major term} of the syllogism, whereas $M$ (resp. $m$) is referred to as \emph{middle term}. 
\end{definition}
The table
\begin{eqnarray}\label{questa}{\textrm{\footnotesize
\begin{tabular}{|l|l|l|l|}
\hline
$\aa MP,\aa SM\therefore\aa SP$&$\aa Pm,\aa SM\therefore
\aa Sp$
&$\ii MP,\aa MS\therefore\ii SP$&$\aa PM,\aa Ms\therefore
\aa Sp$\\
$\aa Mp,\aa SM\therefore\aa Sp$&$\aa PM,\aa Sm\therefore
\aa Sp$
&$\aa MP,\ii MS\therefore\ii SP$&
$\ii PM,\aa MS\therefore\ii SP$\\
$\aa MP,\ii SM\therefore\ii SP$&$\aa Pm,\ii SM\therefore
\ii Sp$
&$\ii Mp,\aa MS\therefore\ii Sp$&
$\aa Pm,\ii MS\therefore\ii Sp$\\
$\aa Mp,\ii SM\therefore\ii Sp$&$\aa PM,\ii Sm\therefore
\ii Sp$
&$\aa Mp,\ii MS\therefore\ii Sp$&\\
\hline
$\aa MP,\aa SM, \ii SS\therefore\ii SP$&$\aa PM,\aa Sm, \ii SS\therefore\ii Sp$&$\aa MP,\aa MS,\ii MM\therefore\ii SP$
&$\aa PM,\aa Ms, \ii SS\therefore\ii Sp$\\
$\aa Mp,\aa SM,\ii SS\therefore\ii Sp$&$\aa Pm,\aa SM,\ii SS\therefore\ii Sp$&$\aa Mp,\aa MS,\ii MM\therefore\ii Sp$
&$\aa Pm,\aa MS,\ii MM\therefore\ii Sp$\\ 
&&&$\aa PM,\aa MS,\ii PP\therefore\ii SP$\\
\hline
\end{tabular}}}
\end{eqnarray}
lists the syllogisms that are known to be valid since Aristotle. The syllogisms in the lower half of the table are the \emph{strengthened} ones, namely those valid under \emph{existential import}, which is the explicit assumption of existence of some $S$, $M$, $P$, $s$, $m$, $p$. For every term-variable $A$ (resp. $a$), the existential import on $A$ (resp. $a$) is explicitly indicated by the occurrence of the categorical proposition $\ii AA$ (resp. $\ii aa$). From the Aristotelian and traditional point of view, existential import on a term-variable $A$ (resp. $a$) was taken for granted and not even mentioned. 

\begin{terminology}
The syllogisms in the upper half of the table~(\ref{questa}) will be henceforth simply referred to as \emph{syllogisms}, whereas those in its lower half will be henceforth referred to as \emph{strengthened syllogisms}.
\end{terminology}

\begin{remark}\label{intanto}
The columns of the table~(\ref{questa}) correspond to the four \emph{figures} into which the valid syllogisms have been divided. From left to right, a syllogism occurring in the first, second, third or fourth column of the table~(\ref{questa}) is said to be in the first, second, third or fourth figure. We will return on the syllogistic figures in section~\ref{redrules}.

The syllogisms listed in the table~(\ref{questa}) are also divided on the base of their \emph{mood}, which is nothing but the order of occurrence of the letters \catebf{A}, \catebf{I} in them. For instance, the mood of the syllogism $\ii MP,\aa MS\therefore\ii SP$ in the third figure, is \catebf{IAI}. 
The mood of the strengthened syllogism $\aa MP, \aa SM, \ii SS\therefore\ii SP$ in the first figure, is \catebf{AAII}. 
\end{remark}

The categorical proposition $\ii{A}{a}$ expresses contradiction. 
The categorical propositions occupy the vertices of the 
\emph{square of opposition}
\begin{eqnarray}\label{sq}
\xymatrix{\aa{A}{B}\ar@{--}[drrr]|{\textrm{contradiction}}
\ar@{-->}[d]_{\textrm{subalternation}}\ar@{--}[rrr]^{\textrm{contrariety}}&&&
\aa{A}{b}\ar@{-->}[d]^{\textrm{subalternation}}\\
\ii{A}{B}\ar@{--}[urrr]|{\hole}
\ar@{--}[rrr]_{\textrm{subcontrariety}}&&&\ii{A}{b}}
\end{eqnarray}
in which:
\begin{itemize}
\item[-] \aa{A}{B} and \ii{A}{b}, as well as \aa{A}{b} and
\ii{A}{B}, are \emph{contradictory} because the arguments $\aa{A}{B},\ii{A}{b}\therefore\ii{A}{a}$ and $\aa{A}{b},\ii{A}{B}\therefore\ii{A}{a}$ are valid.
\item[-] under existential import on $A$, \ii{A}{B} and \ii{A}{b} are \emph{subaltern} to \aa{A}{B} and \aa{A}{b}, respectively, because the arguments $\aa{A}{B},\ii{A}{A}\therefore\ii{A}{B}$ and $\aa{A}{b},\ii{A}{A}\therefore\ii{A}{b}$ are valid.
\item[-] under existential import on $A$, \aa{A}{B} and \aa{A}{b} are \emph{contraries} since the contradictory of each of them follows from the other,  because the arguments $\aa{A}{b},\ii{A}{A}\therefore\ii{A}{b}$ and $\aa{A}{B},\ii{A}{A}\therefore\ii{A}{B}$ are valid.
\item[-] under existential import on $A$, \ii{A}{B} and \ii{A}{b} are \emph{subcontraries} since each of them follows from the contradictory of the other, because the arguments $\aa{A}{B},\ii{A}{A}\therefore\ii{A}{B}$ and $\aa{A}{b}, \ii{A}{A}\therefore\ii{A}{b}$ are valid.
\end{itemize}

\begin{definition}\label{laws}
The previously described valid arguments will be henceforth correspondingly referred to as \emph{laws of contradiction}, \emph{laws of subalternation}, \emph{laws of contrariety}, \emph{laws of subcontrariety}, whereas in their entirety as \emph{laws of the square of opposition}.
\end{definition}

%% file: syllorll.tex
\section{Syllogistics in Rudimentary Linear Logic}\label{syllrll}
We here present a translation of the traditional
syllogistics in a minimal extension of rudimentary linear logic, see~\cite{MR1157804}.
\begin{definition}\label{infrule}
The system \rll, for intuitionistic Rudimentary Linear Logic, has formulas inductively constructed from atomic formulas $\alpha,\beta,\gamma,\ldots$ in accordance with the grammar $A\mathrel{\mathop :}=\alpha\,|\,A\otimes A\,|\, A\multimap A$, where $\otimes$ is \emph{multiplicative conjunction} and $\multimap$ is \emph{linear implication}. A \emph{sequent} in \rll~is a formal expression $\Gamma\vdash A$ where $\Gamma$ is a finite multiset of formulas and $A$ is a formula. The rules of \rll~are 
\[\AxiomC{}
\RightLabel{(Id)}
\UnaryInfC{$A\vdash A$}
\DisplayProof\]
\vspace{.3cm}
\[\AxiomC{$\Gamma,A,B\vdash C$}
\RightLabel{($\otimes$L)}
\UnaryInfC{$\Gamma,A\otimes B\vdash C$}
\DisplayProof
\hspace{3cm}
\AxiomC{$\Gamma\vdash A$}
\AxiomC{$\Delta\vdash B$}
\RightLabel{($\otimes$R)}
\BinaryInfC{$\Gamma,\Delta\vdash A\otimes B$}
\DisplayProof\]
\vspace{.3cm}
\[\AxiomC{$\Gamma\vdash A$}
\AxiomC{$\Delta,B\vdash C$}
\RightLabel{($\multimap$L)}
\BinaryInfC{$\Gamma, \Delta, A\multimap B \vdash C$}
\DisplayProof
\hspace{3cm}
\AxiomC{$\Gamma,A\vdash B$}
\RightLabel{($\multimap$R)}
\UnaryInfC{$\Gamma\vdash A\multimap B$}
\DisplayProof\]
A \emph{proof} in \rll~of a sequent $\Gamma\vdash A$ is a finite tree whose vertices are sequents in \rll, such that leaves are instances of the rule (Id), the root is $\Gamma\vdash A$, and each branching is an instance of a rule of inference. A sequent of \rll~is \emph{provable} if there is a proof for it. Two formulas $A,B$ are \emph{equivalent} if the sequents $A\vdash B$ and $B\vdash A$ are both provable, in which case we write $A\equiv B$. The system $\rll^{\bot}$ is the system that is obtained from \rll~by the addition of the new symbol $\bot$ among the atomic formulas, coherently extending the notion of sequent, proof, and provability, maintaining the same rules.
\end{definition}

\begin{remark}
In~\cite{MR1157804}, the system \rll~was equipped with a cut rule as well. 
We have not included it among the rules in the previous definition because we are not going to use it anywhere in the paper and also because in loc. cit. 
it is proved to be eliminable.
\end{remark}

\begin{notation}\label{botbot}
For every formula $A$ of $\rll^{\bot}$, $A^{\bot}$ is an abbreviation for  $A\multimap\bot$. Despite the symbol $\bot$ there is no reference to \emph{ex falso quodlibet}.
\end{notation}

\begin{definition}\label{dopobot}
The \emph{complement} of a formula $A$ of $\rll^{\bot}$ is the formula $A^{\bot}$.
\end{definition}

\begin{proposition}\label{dopodopo}
For every formulas $A,B,C$ of $\rll^{\bot}$, the following hold:
\begin{itemize}
\item[(i)] the sequent $A\vdash A^{\bot\bot}$ is provable, whereas its converse is not.
\item[(ii)] if $A\vdash B$ is provable, then $B^{\bot}\vdash A^{\bot}$ is provable.
\item[(iii)] the sequent $A\multimap B\vdash B^{\bot}\multimap A^{\bot}$ is provable, wheras its converse is not.
\item[(iv)] the equivalences $A^{\bot}\equiv A^{\bot\bot\bot}$, $A\multimap B^{\bot}\equiv B\multimap A^{\bot}$ and $A\otimes B\equiv B\otimes A$ are provable.
\item[(v)] the sequent $A\otimes A^{\bot}\vdash\bot$ is provable.
\end{itemize}
\end{proposition}
\begin{proof}
Straightforward.
\end{proof}
\begin{definition}\label{cateform}
For every formulas $A$ and $B$ of $\rll^{\bot}$, the formulas
\[\begin{tabular}{llllll}
& $A\multimap B$ &&&
& $A\multimap B^{\bot}$\\
\\
& $A\otimes B$ &&&
& $A\otimes B^{\bot}$
\end{tabular}\]
will be henceforth referred to as \emph{categorical formulas}. In each of them, the subformulas $A$ and $B$ will be henceforth referred to as \emph{subject} and \emph{predicate}, respectively. A \emph{syllogism} in $\rll^{\bot}$ is a sequent $F_1,F_2\vdash C$ where $F_1$, $F_2$ and $C$ are categorical formulas built on three formulas $S$, $M$ and $P$, in the following way: $M$ occurs in both $F_1$, $F_2$ and does not occur in $C$, $S$ occurs in $F_2$ and in $C$, as the subject of the latter, $P$ occurs in $F_1$ and in $C$, as the predicate of the latter. For every formula $A$ of $\rll^{\bot}$, the formula $A\otimes A$ expresses the condition of \emph{existential import} on $A$. For $A\in\{S,M,P\}$, a \emph{strengthened syllogism} in $\rll^{\bot}$ is a sequent $F_1,F_2, A\otimes A\vdash C$ where $F_1,F_2\vdash C$ is a syllogism in $\rll^{\bot}$. 
\end{definition}

\begin{example}
The sequent $M\multimap P, S\multimap M\vdash S\multimap P$ is a syllogism in $\rll^{\bot}$. The sequent $M\multimap P, S\multimap M^{\bot}, M\otimes M\vdash S\multimap P$ is a strengthened syllogism in $\rll^{\bot}$.
\end{example}

\begin{notation}
Whenever confusion is not likely to arise, we will employ the notations
\[\begin{tabular}{llllll}
\aa{A}{B} & &&&
\aa{A}{b} & 
\\
\\
\ii{A}{B} & &&&
\ii{A}{b} & 
\end{tabular}\]
for the corresponding categorical formulas in definition~\ref{cateform}. The sign $\vdash$ in place of the sign $\therefore$ will distinguish a syllogism in $\rll^{\bot}$ from a syllogism as an argument in the natural language. For example, $\aa{P}{M},\ii{S}{M}\vdash\ii{S}{P}$ is a different but equivalent way of writing the syllogism $P\multimap M, S\otimes M\vdash S\otimes P$ in $\rll^{\bot}$.
\end{notation}

\begin{remark}
As already observed in section~\ref{motivation}, the categorical formulas above can be found in~\cite{Abrusci} as well, where they are considered as formulas in the multiplicative fragment of the classical linear logic. We discovered them independently. In loc. cit. the admissibility of the categorical formulas as suitable translations of the corresponding categorical propositions was justified on the base of an intended meaning of the linear logic connectives $\multimap$ and $\otimes$, hinted at in section~\ref{motivation}.
\end{remark}
The next theorem justifies the admissibility of the categorical formulas as a suitable starting point for a reading of the traditional syllogistics within intuitionisitc linear logic in precise mathematical terms.
\begin{theorem}\label{perquesta}
A (strengthened) syllogism is valid if and only if it is provable in $\rll^{\bot}$.
\end{theorem}
\begin{proof}
It can be verified that all the sequents in $\rll^{\bot}$ that correspond to the valid (strengthened) syllogisms in the table~(\ref{questa}) are provable (strengthened) syllogisms in $\rll^{\bot}$.

Conversely we proceed by cases. If $F_1, F_2\vdash C$ is a provable syllogism in $\rll^{\bot}$, then $C$ is a categorical formula among $S\otimes P$, $S\otimes P^{\bot}$, $S\multimap P$,  $S\multimap P^{\bot}$.

If $C$ is $S\otimes P$, then on the base of the inference rules in definition~\ref{infrule}, as a first possibility $F_1$ must be $M\otimes P$ or $P\otimes M$, so that $F_2$ has to be $M\multimap S$. As a second possibility $F_1$ must be $M\multimap P$, so that $F_2$ has to be $S\otimes M$ or $M\otimes S$. Thus, a sequent $F_1,F_2\vdash S\otimes P$ which is a provable syllogism in $\rll^{\bot}$ must validate one of the syllogisms
\[\left.\begin{array}{cccc}
\ii MP, \aa MS\therefore \ii SP&
\ii PM, \aa MS\therefore \ii SP&
\aa MP,\ii SM\therefore\ii SP&
\aa MP,\ii MS\therefore \ii SP
\end{array}\right.\]
all of which occur in table~(\ref{questa}). Similarly, if $C$ is $S\otimes P^{\bot}$.

If $C$ is $S\multimap P$, then on the base of the rules of inference in definition~\ref{infrule}, $F_1$ must be $M\multimap P$ and $F_2$ must be $S\multimap M$, so that a sequent $F_1,F_2\vdash S\multimap P$ which is a provable syllogism must validate the syllogism $\aa{M}{P},\aa{S}{M}\therefore\aa{S}{P}$, which occurs in table~(\ref{questa}). Similarly, if $C$ is $S\multimap P^{\bot}$.

If $F_1,F_2, A\otimes A\vdash C$ is a provable strengthened syllogism in $\rll^{\bot}$, then $C$ is a categorical formula among $S\otimes P$ or $S\otimes P^{\bot}$, because the previous part of the present proof excludes the possibility of having $S\multimap P$ or $S\multimap P^{\bot}$ as conclusions. Thus, as a first case let $C$ be $S\otimes P$ and $A\otimes A$ be $S\otimes S$. The formulas $F_1$ and $F_2$ must be $M\multimap P$ and $S\multimap M$, respectively. Thus a sequent $F_1, F_2, S\otimes S\vdash S\otimes P$ which is a provable strengthened syllogism must validate the syllogism $\aa{M}{P},\aa{S}{M},\ii{S}{S}\therefore\ii{S}{P}$, which occurs in table~(\ref{questa}). The remaining cases are similar.
\end{proof}

The categorical formulas in definition~\ref{cateform} are interrelated by the sequents in $\rll^{\bot}$ which correspond to the laws of the square of opposition, see definition~\ref{laws}. This is the content of 
\begin{proposition}\label{sqseq} 
The following sequents are provable in $\rll^{\bot}$. 
\begin{itemize}
\item[-] laws of contradiction: $A\multimap B,A\otimes B^{\bot}\vdash A\otimes A^{\bot}$ and $A\multimap B^{\bot},A\otimes B\vdash A\otimes A^{\bot}$.
\item[-] laws of subalternation, contrariety and subcontrariety: $A\multimap B, A\otimes A\vdash A\otimes B$ and $A\multimap B^{\bot}, A\otimes A\vdash A\otimes B^{\bot}$.
\end{itemize}
\end{proposition}
\begin{proof}
Straightforward.
\end{proof}
\begin{remark}
The way the laws of contradiction are stated in proposition~\ref{sqseq} amounts to the following: ``conjunctively affirming a pair of categorical formulas corresponding to a pair of contradictory opponents in the square of opposition~(\ref{sq}) leads to a contradicition''. This is so in virtue of point (v) in proposition~\ref{dopodopo}. 
\end{remark}

Another way of expressing a contradictory interrelation between two categorical formulas is the following: ``affirming a categorical formula entails the affirmation of the complement of the categorical formula corresponding to its contradictory opponent in the square of opposition~(\ref{sq}).''. This is the content of 
\begin{proposition}\label{xxx}
The following sequents are provable in $\rll^{\bot}$:
\[\left.\begin{array}{ll}
A\multimap B\vdash (A\otimes B^{\bot})^{\bot}\\
\\
A\otimes B\vdash (A\multimap B^{\bot})^{\bot}
\end{array}\right.
\qquad
\qquad
\left.\begin{array}{ll}
A\multimap B^{\bot}\vdash (A\otimes B)^{\bot}\\
\\
A\otimes B^{\bot}\vdash (A\multimap B)^{\bot}
\end{array}\right.\]
\end{proposition}
\begin{proof}
Straightforward.
\end{proof}

%% file: catform.tex
\section{A diagrammatic calculus of syllogisms}\label{sec4}
We here recall some of the features of the diagrammatic formal system SYLL for the traditional syllogistics that we introduced in~\cite{Pagnan2012-PAGADC}, and point out its connection with the reading of the traditional syllogistics in intuitionistic linear logic that we described in section~\ref{syllrll}. 

In this section, the material which is strictly related to SYLL is essentially from~\cite{Pagnan2012-PAGADC}, to which we refer the reader for further details.

We prove that a (strengthened) syllogism in $\rll^{\bot}$ is provable there if and only if it is diagrammatically provable in (an extension of) SYLL. This is the main novelty in this section.

\begin{definition}\label{synt}
The \emph{diagrammatic primitives} of SYLL are the arrow and bullet symbols $\arr$, $\leftarrow$, $\b$. The \emph{linguistic primitives} of SYLL consist of  term-variables $A,B,C,\ldots$. The \emph{syntactic primitives} of SYLL are the diagrammatic or linguistic primitives. A \emph{diagram} of SYLL, or SYLL \emph{diagram}, is a finite list of arrow symbols separated by a single bullet symbol or term-variable, beginning and ending at a term-variable. The \emph{reversal} of a SYLL diagram is the diagram obtained by specular symmetry. A \emph{part} of a SYLL diagram is a finite list of consecutive components of a diagram. 
\end{definition}

\begin{example}
The lists $A\arr X$, $A\leftarrow A$, $A\arr\b\arr B$, $X\arr Y\arr\b\leftarrow X$ are examples of SYLL diagrams. Their reversals are $X\leftarrow A$, $A\arr A$, $B\leftarrow\b\leftarrow A$ and $X\arr\b\leftarrow Y\leftarrow X$, respectively.
Every SYLL diagram is a part of itself and in general, a part of a SYLL diagram need not be a SYLL diagram. The lists $A\arr$, $\leftarrow$, $\b\arr B$, $Y\arr\b$ are not SYLL diagrams. Nonetheless, they are parts of the SYLL diagrams above.
The lists $\b\b$, $A\b\leftarrow$, $A\b X$ are neither SYLL diagrams nor parts of SYLL diagrams.
\end{example}

\begin{notation}\label{parts}
Parts of diagrams will be henceforth denoted by calligraphic upper case letters such as $\cate{D},\cate{E}$, etc. In order to distinguish explicitly a part with respect to a whole diagram, we adopt a heterogeneous notation mixing calligraphic upper case letters and syntactic primitives of SYLL. For example and future reference, see definition~\ref{speriamo}, the writing $\cate{D}\b\leftarrow A\arr\b\cate{E}$ refers to a diagram in which the part $\b\leftarrow A\arr\b$ has been distinguished with respect to the remaining parts \cate{D} and \cate{E}. Thus, it may be the case that the whole diagram looks like $X\leftarrow\b\leftarrow A\arr\b\arr Y$ or $S\arr\b\arr\b\leftarrow A\arr\b\leftarrow S$, so that the part \cate{D} would be $X\leftarrow$ or $S\arr\b\arr$, whereas the part \cate{E} would be $\arr Y$ or $\leftarrow S$, in each respective case.
\end{notation}

\begin{definition}\label{compl}
Let \cate{D} be a SYLL diagram.  A part \cate{A} of \cate{D} occurs in 
\emph{complemented mode} in \cate{D}, if the list $\cate{A}\arr\b$ (resp. $\b\leftarrow\cate{A}$) is a part of \cate{D}. 
\end{definition}
\begin{remark}
A term-variable $A$ in \cate{D} occurs in \emph{complemented mode} in \cate{D}, if the list $A\arr\b$ (resp. $\b\leftarrow A$) is a part of \cate{D}.
\end{remark}

\begin{example}\label{pxx}
A term-variable $M$ occurring in a diagram \cate{D} like $\cate{A}\b\leftarrow M\arr\b\cate{B}$ occurs in complemented mode twice, because the lists $M\arr\b$ and $\b\leftarrow M$ are both parts of \cate{D}. This state of things will turn out to be in connection with the occurrence of $M$ as a middle term-variable in complemented mode in both the premises of a valid syllogism with complements of terms, see remark~\ref{uupp}. Moreover, the part $\cate{A}\b\leftarrow M$ occurs in complemented mode in \cate{D}, because the list $\cate{A}\b\leftarrow M\arr\b$ is a part of \cate{D}.
In a SYLL diagram like $A\arr\b\arr\b\leftarrow B$ the term-variables $A$ and $B$ both occur in complemented mode. In the same diagram, the part $A\arr\b$ occurs in complemented mode. One would be tempted to consider the occurrence of $A$ in the previous diagram as $A$ in complemented mode twice, in connection with the complement of a formula $A^{\bot}$, that is $A^{\bot\bot}$, see definition~\ref{dopobot}.  We will be able to make sense of this in remark~\ref{primaopoi} and in remark~\ref{eco}. Similarly for the occurrence of $A$ in a part such as $\b\leftarrow\b\leftarrow A$.
\end{example}

\begin{definition}\label{sylldiagr}
The SYLL diagrams
\[\begin{tabular}{lllll}
& \AA{A}{B} &&
& \EE{A}{B}\\
\\
& \II{A}{B} &&
& \OO{A}{B}
\end{tabular}\]
for the corresponding categorical propositions will be henceforth referred to as \emph{syllogistic diagrams}. Their reversals are
\[\begin{tabular}{llll}
\Ao{A}{B} && \Eo{A}{B}\\
\\
\Io{A}{B} && \Oo{A}{B}
\end{tabular}\]
respectively.
\end{definition}

\begin{remark}\label{ascanso}
In every syllogistic diagram the term-variable $A$, which is intended as the subject of the corresponding categorical proposition, does not occur in complemented mode.
Similarly for the term-variables $B$ occurring in the syllogistic diagrams for $\aa AB$ and $\ii AB$. In the syllogistic diagram for \ii{A}{b} the term-variable $B$ occurs in complemented mode, because of its occurrence in the part $\b\leftarrow B$. Similarly for the term-variable $B$ in the diagram for \aa Ab. The reversals of the syllogistic diagrams for \aa{A}{b} and \ii{A}{B} are those for \aa{B}{a}, in which $A$ occurs in complemented mode, and \ii{B}{A}. The correspondence between them 
and the categorical formulas $B\multimap A^{\bot}$ and $B\otimes A$, respectively, is apparent and coherent with their symmetric look and the provable equivalences $A\multimap B^{\bot}\equiv B\multimap A^{\bot}$ and $A\otimes B\equiv B\otimes A$, see point (iv) of proposition~\ref{dopodopo}. The reversal of the syllogistic diagram for \ii{A}{b} corresponds to the formula $B^{\bot}\otimes A$.
\end{remark}

\begin{definition}
A \emph{concatenable pair} of SYLL diagrams is a pair of SYLL diagrams $(\cate{D}A,A\cate{E})$ or $(A\cate{E},\cate{D}A)$ whose \emph{concatenation} is, in both cases, the SYLL diagram $\cate{D}A\cate{E}$ which is obtained by overlapping its components on the common extremal term-variable $A$. A \emph{composable pair} of SYLL diagrams is a concatenable pair $(\cate{D}\arr A,A\arr\cate{E})$ or $(A\arr\cate{E},\cate{D}\arr A)$, $(\cate{D}\leftarrow A,A\leftarrow\cate{E})$ or $(A\leftarrow\cate{E},\cate{D}\leftarrow A)$. In the first two cases, a composable pair gives rise to a \emph{composite} SYLL diagram
$\cate{D}\arr\cate{E}$ obtained by substituting the part $\arr A\arr$ in the concatenation $\cate{D}\arr A\arr \cate{E}$ with the sole, accordingly oriented, arrow symbol $\arr$. Analogously, in the second two cases, a composable pair gives rise to a composite SYLL diagram $\cate{D}\leftarrow\cate{E}$. 
\end{definition}

\begin{examples}
For every term-variable $A$, $(A,A)$ is a concatenable pair whose concatenation is the diagram $A$. It is not a composable pair since no arrow symbols occur. The pair $(A\leftarrow B, X\arr B)$ is not concatenable, thus not composable, whereas the pair $(A\leftarrow B,B\leftarrow X)$ is concatenable and composable, with concatenation $A\leftarrow B\leftarrow X$ and 
composite $A\leftarrow X$. The pair $(X\arr A, A\leftarrow B)$ is concatenable but not composable.
The pair $(X\leftarrow B,B\leftarrow X)$ is concatenable in two different ways by overlapping its components either on $B$ or on $X$. Also, it is composable in two different ways giving rise to either the composite $X\leftarrow X$ or the composite $B\leftarrow B$, respectively. 
\end{examples}

\begin{definition}\label{wfd}
A \emph{well-formed diagram} of SYLL is defined inductively as follows:
\begin{itemize}
\item[(i)] a syllogistic diagram is a well-formed diagram.
\item[(ii)] the reversal of a syllogistic diagram is a well-formed diagram.
\item[(iii)] a diagram which is the concatenation of a concatenable pair whose components are 
well-formed diagrams is a well-formed diagram.
\end{itemize}
\end{definition}
The informal usage of the syllogistic diagrams and their reversals for the verification of the validity of syllogisms is as follows: given a syllogism, one considers the three syllogistic diagrams or their reversals to represent its first premise, second premise and conclusion. These involve three distinguished term-variables, usually denoted $S$, $P$ and $M$, so that $M$ occurs in both the diagrams for the premises, to let them be concatenable, and does not in the conclusion, whereas $S$ and $P$ occur in the conclusion as well as in the premises. Verifying the validity of a syllogism consists in calculating the composite diagram of the concatenation of its premises, if these form a composable pair, and compare it with the diagram for the conclusion. For example, the validity of the syllogism $\aa{P}{M},\aa{S}{m}\therefore\aa{S}{p}$ is suggestively verified through a drawing such as
\[\xymatrix@C=7ex{S\ar@(r,l)[dr]&&M\ar@(l,r)[dl]&P\ar[l]\ar@(d,r)[dll]\\
&\b&&}\]
in which, the curved arrow from $P$ to $\b$ has been obtained by ``composing'' the accordingly oriented arrows through $M$, so to obtain in the end a direct connection between $S$ and $P$ through the syllogistic diagram for \aa{S}{p}, which is the one for the conclusion of the syllogism in consideration. More succintly, such an informal procedure will be henceforth written as a single line of inference like 
\begin{eqnarray}\label{sylloline}
\AxiomC{$S\arr\b\leftarrow M\leftarrow P$}
\UnaryInfC{\EE{S}{P}}
\DisplayProof
\end{eqnarray}
The invalidity of the syllogism $\ii{P}{m},\aa{M}{s}\therefore\ii{S}{P}$ is confirmed by the fact that the pair $(P\leftarrow\b\arr\b\leftarrow M,M\arr\b\leftarrow S)$ although concatenable is not composable. In other words,
the line of inference
\[\AxiomC{$S\arr\b\leftarrow M\arr\b\leftarrow\b\arr P$}
\UnaryInfC{}
\DisplayProof\]
displays no conclusion.

\begin{remark}\label{doppo}
In calculating the composite of a composable pair of diagrams no bullet symbol is deleted, so that the composite contains as many bullets as in the concatenation of the diagrams in the given pair. This
is useful, when one also takes into account the orientation of the involved arrow symbols, for rejecting an invalid form of syllogism. For instance, the syllogism $\ii{P}{m},\aa{M}{s}\therefore\ii{S}{P}$ is invalid since a single bullet symbol occurs in the syllogistic diagram for the conclusion, whereas three of them occur in those for the premises. The syllogism $\aa{P}{M},\ii{S}{M}\therefore \aa{S}{p}$ is invalid since the syllogistic diagram for the conclusion contains a single bullet and a pair of arrows converging to it, whereas a single bullet and a pair of arrows diverging from it are contained in the syllogistic diagram for the second premise.
\end{remark}
For every term-variable $A$, interesting instances of syllogistic diagrams are:
\[\begin{tabular}{llllll}
& \AA{A}{A} && & \EE{A}{A}\\
\\
& \II{A}{A} && & \OO{A}{A}
\end{tabular}\]
where the diagrams for 
$\catebf{A}_{AA}$ and $\catebf{I}_{AA}$ correspond to the \emph{laws of identity}, see~\cite{Lukasiewicz}. In particular, the diagram for $\catebf{I}_{AA}$ expresses existential import on the term-variable $A$ in SYLL. 
For example, the validity of the strengthened syllogism $\aa MP, \aa SM, \ii SS\therefore \ii SP$ is informally expressed by the line of inference
\[\AxiomC{$S\leftarrow\b\arr S\arr M\arr P$}
\UnaryInfC{$S\leftarrow\b\arr P$}
\DisplayProof\]
which displays the deletion of both the term variables $S$ and $M$. 

The diagram for $\aa Aa$ expresses the ``emptyness'' of the term-variable $A$, whereas the diagram for $\catebf{I}_{Aa}$ is an expression of contradiction. 
\begin{remark}\label{diagsq}
All the laws of the square of opposition~(\ref{sq}) are easily calculated through the employment of the syllogistic diagrams and their reversals. For example, the contradiction between \aa{A}{B} and \ii{A}{b}, that is $\aa{A}{B},\ii{A}{b}\therefore\ii{A}{a}$, is informally verified through
\[\AxiomC{$A\leftarrow\b\arr\b\leftarrow B\leftarrow A$}
\UnaryInfC{$A\leftarrow\b\arr\b\leftarrow A$} 
\DisplayProof\]
whereas the contrariety of \aa{A}{B} and \aa{A}{b}, one of which is $\aa{A}{B},\ii{A}{A}\therefore\ii{A}{B}$, is verified through
\[\AxiomC{$A\leftarrow\b\arr A\arr B$}
\UnaryInfC{$A\leftarrow\b\arr B$}
\DisplayProof\]
We leave the easy verification of the remaining laws as an exercise for the reader.
\end{remark}
\begin{definition}\label{rules}
The rules of inference of SYLL are the following:
\[\begin{tabular}{ccccc}
\AxiomC{\AA{A}{B}}\doubleLine
\UnaryInfC{\Ao{A}{B}}\DisplayProof
&&\AxiomC{\EE{A}{B}}\doubleLine
\UnaryInfC{\Eo{A}{B}}\DisplayProof&
\end{tabular}\]
\[\begin{tabular}{ccccc}
\AxiomC{\II{A}{B}}\doubleLine
\UnaryInfC{\Io{A}{B}}\DisplayProof
&&\AxiomC{\OO{A}{B}}\doubleLine
\UnaryInfC{\Oo{A}{B}}\DisplayProof&
\end{tabular}\]
\vspace{.5cm}
\[\begin{tabular}{ccccc}
\AxiomC{\cate{D}A}
\AxiomC{A\cate{E}}
\BinaryInfC{\cate{D}A\cate{E}}
\DisplayProof&&
\AxiomC{A\cate{E}}
\AxiomC{\cate{D}A}
\BinaryInfC{\cate{D}A\cate{E}}
\DisplayProof&
\end{tabular}\]
\[\begin{tabular}{ccccc}
\AxiomC{$\cate{D}\arr A\arr\cate{E}$}
\UnaryInfC{$\cate{D}\arr\cate{E}$}
\DisplayProof&&
\AxiomC{$\cate{D}\leftarrow A\leftarrow\cate{E}$}
\UnaryInfC{$\cate{D}\leftarrow\cate{E}$}
\DisplayProof&
\end{tabular}\]
where the double line means that the rule can be used top-down as well as bottom-up. A \emph{sequent} of SYLL is a formal expression $\cate{G}_1,\ldots,\cate{G}_n\models\cate{A}$ where the $\cate{G}_i$'s are well-formed
diagrams of SYLL any two consecutive of which, or their reversals, are concatenable and \cate{A} is a well-formed diagram of SYLL. A \emph{syllogism} in SYLL is a sequent $\cate{D}_1,\cate{D}_2\models\cate{C}$, where $\cate{D}_1$, $\cate{D}_2$ and \cate{C} are syllogistic diagrams in which exactly three term-variables $S$, $M$ and $P$ occur, in the following way: $M$ occurs in both the diagrams for the premises, and does not in \cate{C}, $S$ occurs in $\cate{D}_2$, $P$ occurs in $\cate{D}_1$ and both occur in \cate{C}. A \emph{proof} in SYLL of a sequent $\cate{G}_1,\ldots,\cate{G}_n\models\cate{A}$ is a finite tree where each node is a well-formed diagram of SYLL, the root is \cate{A}, the leaves are the $\cate{G}_i$'s and each branching is an instance of a rule of inference. A sequent is \emph{provable} in SYLL if there is a proof for it.
\end{definition}

\begin{example}
The sequent $\II XA, \AA XY, \EE YB\models \OO AB$ is provable in SYLL, since
\[\AxiomC{\II XA}
\UnaryInfC{\II AX}
\AxiomC{\AA XY}
\AxiomC{\EE YB}
\BinaryInfC{$X\arr Y\arr\b\leftarrow B$}
\UnaryInfC{\EE XB}
\BinaryInfC{$A\leftarrow\b\arr X\arr\b\leftarrow B$}
\UnaryInfC{\OO AB}
\DisplayProof\]
is a proof in SYLL for it.
\end{example}
\begin{notation}
Whenever confusion is not likely to arise, we will employ the notations
\[\begin{tabular}{llllll}
\aa{A}{B} & &&&
\aa{A}{b} & 
\\
\\
\ii{A}{B} & &&&
\ii{A}{b} & 
\end{tabular}\]
for the corresponding syllogistic diagrams in definition~\ref{sylldiagr}. The sign $\models$ in place of the sign $\therefore$ will distinguish a syllogism in SYLL from a syllogism as an argument in the natural language. For example, $\aa{P}{M},\ii{S}{M}\models\ii{S}{P}$ is a different but equivalent way of writing the syllogism $\AA PM, \II SM \models \II SP$ in SYLL.
\end{notation}

\begin{example}
The syllogism $\aa{P}{M},\aa{S}{m}\models\aa{S}{p}$ is provable in SYLL, since 
\[\AxiomC{\EE{S}{M}}
\AxiomC{\AA{P}{M}}
\UnaryInfC{\Ao{P}{M}}
\BinaryInfC{$S\arr\b\leftarrow M\leftarrow P$}
\UnaryInfC{\EE{S}{P}}
\DisplayProof\]
is a proof in SYLL for it.
\end{example}

\begin{remark}
Because of the last couple of rules in definition~\ref{rules}, proving sequents of SYLL is a resource-sensitive activity, in the sense that the salient branchings in a proof of a sequent consist of the elimination of exactly one mediating term-variable, intuitively corresponding to the exhaustion of an available resource. This is in line with the intended interpretation of linear logic as a resource-sensitive logic, hinted at in section~\ref{motivation}. See~\cite{MR1356006}.
\end{remark}

\begin{theorem}\label{propopropo}
A syllogism is valid if and only if it is provable in SYLL.
\end{theorem}
\begin{proof}
See~\cite{Pagnan2012-PAGADC}.
\end{proof}

\begin{definition}\label{syllp}
Let $\textup{SYLL}^+$ denote the formal system which is obtained from SYLL by the addition of the rule
\[\AxiomC{}
\UnaryInfC{\II{A}{A}}
\DisplayProof\]
to the rules in definition~\ref{rules}, with suitable extension of the remaining notions introduced there.
\end{definition}

\begin{theorem}\label{propos}
A strengthened syllogism is valid if and only if it is provable in $\textup{SYLL}^+$.
\end{theorem}
\begin{proof}
See~\cite{Pagnan2012-PAGADC}.
\end{proof}
\begin{example}
The strengthend syllogism $\aa Mp,\aa MS, \ii MM\models\ii Sp$ is provable in $\textup{SYLL}^+$, since
\[\AxiomC{\AA{M}{S}}
\UnaryInfC{\Ao{M}{S}}
\AxiomC{}
\UnaryInfC{\II{M}{M}}
\AxiomC{\EE{M}{P}}
\BinaryInfC{$M\leftarrow\b\arr M\arr\b\leftarrow P$}
\UnaryInfC{\OO{M}{P}}
\BinaryInfC{$S\leftarrow M\leftarrow\b\arr\b\leftarrow P$}
\UnaryInfC{\OO{S}{P}}
\DisplayProof\]
is a proof in SYLL for it.
\end{example}

\begin{theorem}\label{ddd}
A (strengthened) syllogism is provable in $\rll^{\bot}$ if and only if it is provable in $\textup{SYLL}^{+}$.
\end{theorem}
\begin{proof}
The result follows from the theorems~\ref{perquesta},~\ref{propopropo},~\ref{propos}.
\end{proof}

%% file: redrules.tex
\section{The reduction rules of the syllogistics}\label{redrules}
In this section we comment on the so-called \emph{reduction rules} for the syllogistics in connection with the diagrammatic calculus we presented in section~\ref{sec4}. The material which is in this section is not contained in~\cite{Pagnan2012-PAGADC}.

The reduction rules for the syllogistics were introduced by Aristotle to bring back
the valid syllogisms in the second, third and fourth figures, see remark~\ref{intanto}, to the valid syllogisms in the first figure, to let these acquire the privileged status of \emph{perfect syllogisms} as named by Aristotle himself. Reducing a syllogism is a way of justifying it as an accepatable valid inference schema. The reduction rules are the following:
\begin{itemize}
\item[-] \emph{simple conversion}:  apply the valid arguments $\aa{A}{b}\therefore\aa{B}{a}$ and $\ii{A}{B}\therefore\ii{B}{A}$ to a premise or to the conclusion of a syllogism.
\item[-] \emph{conversion per accidens}: apply the arguments $\aa{B}{A}, \ii BB\therefore\ii{A}{B}$ and $\aa{B}{a}, \ii AA\therefore\ii{A}{b}$ to a premise or to the conclusion of a syllogism. 
\item[-] \emph{subalternation}: apply the arguments $\aa{A}{B}, \ii AA\therefore\ii{A}{B}$ and $\aa{A}{b}, \ii AA\therefore\ii{A}{b}$ to a premise or to the conclusion of a syllogism. 
\item[-] \emph{exchange of premises}: it consists in exchanging the order of the premises in a syllogism, simultaneously letting the minor term become the major and viceversa.
\item[-] \emph{contradiction}: assume the premises and the contradictory of the conclusion of a valid syllogism in the second, third or fourth figure and obtain the contradictory of one of the premises by means of a syllogism in the first figure.
\end{itemize}
The previous rules are connected with the valid $1$-term and $2$-term syllogisms
whose description follows. The general $n$-term syllogisms are discussed in~\cite{Lukasiewicz, Meredith, MR0107598}. The valid $1$-term syllogisms are exactly two, that is $\vdash\catebf{A}_{AA}$ and $\vdash\catebf{I}_{AA}$. These are the laws of identity we previously hinted at. The valid $2$-term syllogisms are ten in total:
\begin{description}
\item[] $\catebf{A}_{AB}\vdash\catebf{A}_{AB}$
\item[] $\catebf{A}_{Ab}\vdash\catebf{A}_{Ab}$
\item[] $\catebf{I}_{AB}\vdash\catebf{I}_{AB}$
\item[] $\catebf{I}_{Ab}\vdash\catebf{I}_{Ab}$
\item[] the \emph{laws of subalternation}:
\begin{description}
\item[] $\catebf{A}_{AB}, \ii AA\vdash\catebf{I}_{AB}$
\item[] $\catebf{A}_{Ab},\ii AA\vdash\catebf{I}_{Ab}$.
\end{description}
\item[] the \emph{laws of simple conversion}:
\begin{description}
\item[] $\catebf{A}_{Ab}\vdash\catebf{A}_{Ba}$
\item[] $\catebf{I}_{AB}\vdash\catebf{I}_{BA}$ 
\end{description}
\item[] the \emph{laws of conversion per accidens}:
\begin{description}
\item[] $\catebf{A}_{BA}, \ii BB\vdash\catebf{I}_{AB}$
\item[] $\catebf{A}_{Ba}, \ii AA\vdash\catebf{I}_{Ab}$
\end{description}
\end{description}

\begin{definition}
Let $\textup{SYLL}^{++}$ be the formal system which is obtained from $\textup{SYLL}^{+}$ by the addition of the rule
\[\AxiomC{}
\UnaryInfC{\AA{A}{A}}
\DisplayProof\]
together with suitable extension of the remaining notions introduced in definition~\ref{syllp}.
\end{definition}
In $\textup{SYLL}^{++}$ it is possible to prove all the valid $1$-term and $2$-term syllogisms. Most of them are immediate, such as for example the laws of identity. Here is the proofs of the laws of conversion per accidens:
\begin{itemize}
\item[-] $\catebf{A}_{BA}, \ii BB\vdash\catebf{I}_{AB}$
\[\AxiomC{}
\UnaryInfC{\II{B}{B}}
\AxiomC{\AA{B}{A}}
\BinaryInfC{$B\leftarrow\b\arr B\arr A$}
\UnaryInfC{\II{B}{A}}
\UnaryInfC{\II{A}{B}}
\DisplayProof\]
\item[-] $\catebf{A}_{Ba}, \ii AA\vdash\catebf{I}_{Ab}$
\[\AxiomC{}
\UnaryInfC{\II{A}{A}}
\AxiomC{\EE{B}{A}}
\UnaryInfC{\EE{A}{B}}
\BinaryInfC{$A\leftarrow\b\arr A\arr\b\leftarrow B$}
\UnaryInfC{\OO{A}{B}}
\DisplayProof\]
\end{itemize}
In $\textup{SYLL}^{++}$, $\aa{P}{m},\ii{M}{S}\vdash\ii{S}{p}$ reduces to $\aa{M}{p},\aa{S}{M}\vdash\aa{S}{p}$ by contradiction and simple conversion:
\[\AxiomC{\AA{S}{P}}
\AxiomC{\EE{P}{M}}
\BinaryInfC{$S\arr P\arr\b\leftarrow M$}
\UnaryInfC{\EE{S}{M}}
\UnaryInfC{\EE{M}{S}}
\DisplayProof\]
In fact, from the leaves to the diagram \EE SM, for \aa{S}{m}, the previous proof tree is an instance of the second syllogism, which is in the first figure, with premises \aa{S}{P}, which is the contradictory of \ii{S}{p}, and \aa{P}{m}, which is the first premise of the first syllogism. The root is obtained by simple conversion on \aa{S}{m} towards \aa{M}{s}, which is the contradictory of the second premise of the first syllogism. Similarly, the syllogism $\aa PM, \ii Sm \vdash \ii Sp$ reduces to the syllogism $\aa PM, \aa SM \vdash \aa SP$ by contradiction:
\[\AxiomC{\AA SP}
\AxiomC{\AA PM}
\BinaryInfC{\AA SM}
\DisplayProof\]
The strengthened syllogism $\aa{P}{M},\aa{M}{s}, \ii SS\vdash\ii{S}{p}$ reduces to the strengthend syllogism $\aa{M}{p},\aa{S}{M}, \ii SS\vdash\ii{S}{p}$ by exchanging the premises, and suitably renaming the term-variables $s$ and $P$. The strengthened syllogism $\aa{M}{P},\aa{S}{M}, \ii SS\vdash\ii{S}{P}$ reduces to the syllogism $\aa{M}{P},\aa{S}{M}\vdash\aa{S}{P}$ by subalternation:
\[\AxiomC{}
\UnaryInfC{\II{S}{S}}
\AxiomC{\AA{S}{M}}
\AxiomC{\AA{M}{P}}
\BinaryInfC{$S\arr M\arr P$}
\UnaryInfC{\AA{S}{P}}
\BinaryInfC{$S\leftarrow\b\arr S\arr P$}
\UnaryInfC{\II{S}{P}}
\DisplayProof\]

%% file: demorgan.tex
\section{Extension to De Morgan's syllogistics}\label{demorgan}
In connection with~\cite{MR2464674} and~\cite{ADeMorgan}, we mention that the investigation toward the possibility of considering categorical propositions with complemented subjects, see section~\ref{syllru}, lead to an extension of traditional syllogistic. In this section we extend the diagrammatic system SYLL for the purpose, together with theorems~\ref{propopropo} and~\ref{propos}. All of this is not in~\cite{Pagnan2012-PAGADC}. Moreover, we extend theorem~\ref{perquesta} and theorem~\ref{ddd} as well.

The starting point are the four \emph{new categorical propositions}
\[\begin{tabular}{lllll}
&\aa{a}{B}: Each non-$A$ is $B$ &&\aa{a}{b}: Each non-$A$ is non-$B$\\
\\
\\
&\ii{a}{B}: Some non-$A$ is $B$ &&\ii{a}{b}: Some non-$A$ is non-$B$
\end{tabular}\]
all of which are affirmative, universal or particular. A prerogative of De Morgan's approach to syllogistics was in fact that of making the affirmative mode of predication the fundamental one. Moreover, De Morgan introduced the so-called \emph{spicular notation} for a symbolic treatment of syllogistics and on the base of it and the available eight categorical propositions, he was able to find thirty-two valid syllogisms, eight in the mood \textbf{AAA}, eight in the mood \textbf{AII}, eight in the mood \textbf{IAI} and eight strengthened ones in the mood \textbf{AAII}. A comparison between the diagrammatic system SYLL and the spicular notation is treated in~\cite{Pagnan}.
\begin{definition}
A \emph{De Morgan syllogism} is an argument $P_1,P_2\therefore C$ where $P_1,P_2,C$ are (new) categorical propositions, in which three term-variables among $S,M,P,s,m,p$ occur as follows: $M$ (resp. $m$) occurs in both the premises and does not occur in the conclusion whereas $P$ (resp. $p$) occurs in the first premise and $S$ (resp. $s$) occurs in the second premise. The term-variables $S$ (resp. $s$) and $P$ (resp. $p$) occur as the subject and predicate of the conclusion, respectively, and are referred to as \emph{minor term} and \emph{major term} of the syllogism, whereas $M$ (resp. $m$) is referred to as \emph{middle term}. A \emph{strengthened De Morgan syllogism} is an argument $P_1,P_2,\ii XX\therefore C$ or $P_1,P_2,\ii xx\therefore C$, where $X\in\{S,M,P\}$, $x\in\{s,m,p\}$ and $P_1,P_2\therefore C$ is a De Morgan syllogism.
\end{definition}
The table 
\begin{eqnarray}\label{questat}{\textrm{\footnotesize
\begin{tabular}{|l|l|l|l|}
\hline
$\aa MP,\aa SM\therefore \aa SP $& $\aa MP,\ii SM\therefore \ii SP $& $\ii MP,\aa MS\therefore\ii SP$ & $\aa MP,\aa SM,\ii SS\therefore \ii SP$\\
$\aa mP, \aa Sm\therefore\aa SP$ & $\aa mP, \ii Sm\therefore\ii SP$ & $\ii mP,\aa mS\therefore \ii SP$ & $\aa mP,\aa Sm, \ii SS\therefore \ii SP$\\
$\aa Mp, \aa SM\therefore \aa Sp$&$\aa Mp, \ii SM\therefore \ii Sp$ & $\ii MP,\aa Ms\therefore \ii sP$ & $\aa Mp,\aa SM,\ii SS\therefore \ii Sp$\\
$\aa mp,\aa Sm\therefore \aa Sp$&$\aa mp,\ii Sm\therefore\ii Sp$ &$\ii mP,\aa ms\therefore\ii sP$&$\aa mp,\aa Sm,\ii SS\therefore \ii Sp$\\
$\aa MP, \aa sM\therefore\aa sP$ &$\aa MP,\ii sM\therefore \ii sP$&$\ii Mp,\aa MS\therefore\ii Sp$&$\aa MP,\aa sM,\ii ss\therefore \ii sP$\\
$\aa mP,\aa sm\therefore\aa sP$&$\aa mP,\ii sm\therefore\ii sP$&$\ii mp,\aa mS\therefore\ii Sp$&$\aa mP,\aa sm,\ii ss\therefore \ii sP$\\
$\aa Mp,\aa sM \therefore \aa sp$&$\aa Mp,\ii sM\therefore\ii sp$&$\ii Mp,\aa Ms\therefore \ii sp$&$\aa Mp,\aa sM,\ii ss\therefore \ii sp$\\
$\aa mp,\aa sm\therefore \aa sp$&$\aa mp,\ii sm\therefore\ii sp$&$\ii mp, \aa ms\therefore\ii sp$&$\aa mp,\aa sm,\ii ss\therefore \ii sp$\\
\hline
\end{tabular}}}
\end{eqnarray}
lists exactly the valid De Morgan syllogisms. It was filled in on the base of an analogous one contained in~\cite{MR2464674}.

\begin{definition}
For every formulas $A$ and $B$ in $\rll^{\bot}$, the formulas 
\[\begin{tabular}{llllll}
$A^{\bot}\multimap B$&&&&
$A^{\bot}\multimap B^{\bot}$\\
\\
$A^{\bot}\otimes B$&&&&
$A^{\bot}\otimes B^{\bot}$
\end{tabular}\]
will be henceforth referred to as \emph{new categorical formulas}. Correspondingly, the SYLL diagrams
\[\begin{tabular}{llllll}
&$A\arr\b\arr B$&&&
&$A\arr\b\arr\b\leftarrow B$\\
\\
&$A\arr\b\leftarrow\b\arr B$&&&
&$A\arr\b\leftarrow\b\arr\b\leftarrow B$
\end{tabular}\]
will be henceforth referred to as \emph{new syllogistic diagrams}. Their reversals are
\[\begin{tabular}{llllll}
&$B\leftarrow\b\leftarrow A$&&&
&$B\arr\b\leftarrow\b\leftarrow A$\\
\\
&$B\leftarrow\b\arr\b\leftarrow A$&&&
&$B\arr\b\leftarrow\b\arr\b\leftarrow A$
\end{tabular}\]
respectively.
\end{definition}
\begin{notation}
When confusion is not likely to arise we will employ the notations
\[\begin{tabular}{llllll}
\aa{a}{B}&  &&&
\aa{a}{b} &\\
\\
\ii{a}{B} &  &&&
\ii{a}{b} & 
\end{tabular}\]
for the corresponding new categorical formulas and new syllogistic diagrams. 
\end{notation}
The next result extends proposition~\ref{sqseq}.
\begin{proposition}\label{nsq}
The following sequents are provable in $\rll^{\bot}$:
\begin{itemize}
\item[-] (new) laws of contradiction: $A^{\bot}\multimap B,A^{\bot}\otimes B^{\bot}\vdash B\otimes B^{\bot}$ and $A^{\bot}\multimap B^{\bot},A^{\bot}\otimes B\vdash B\otimes B^{\bot}$.
\item[-] (new) laws of subalternation, contrariety and subcontrariety: $A^{\bot}\multimap B, A^{\bot}\otimes A^{\bot}\vdash A^{\bot}\otimes B$ and $A^{\bot}\multimap B^{\bot},A^{\bot}\otimes A^{\bot}\vdash A^{\bot}\otimes B^{\bot}$.
\end{itemize}
\end{proposition}
\begin{proof}
Straightforward.
\end{proof}
The next result extends proposition~\ref{xxx}
\begin{proposition}
The following sequents are provable in $\rll^{\bot}$:
\[\left.\begin{array}{ll}
A^{\bot}\multimap B\vdash (A^{\bot}\otimes B^{\bot})^{\bot}\\
\\
A^{\bot}\otimes B\vdash (A^{\bot}\multimap B^{\bot})^{\bot}
\end{array}\right.
\qquad
\qquad
\left.\begin{array}{ll}
A^{\bot}\multimap B^{\bot}\vdash (A^{\bot}\otimes B)^{\bot}\\
\\
A^{\bot}\otimes B^{\bot}\vdash (A^{\bot}\multimap B)^{\bot}
\end{array}\right.\]
\end{proposition}
\begin{proof}
Straightforward.
\end{proof}

\begin{remark}\label{gh}
In the multiplicative fragment of classical linear logic the 
new categorical formula $A^{\bot}\multimap B$ would properly correspond to the multiplicative disjunction of $A$ and $B$, $A\parr B$, in the sense that one would be allowed to define it as $A^{\bot}\multimap B$ and then be able to prove
$A\parr B\equiv B\parr A$, using that in classical linear logic $X\equiv X^{\bot\bot}$, for every formula $X$. In the present framework such equivalence is not provable, point (i) of proposition~\ref{dopodopo},  and the sequent $A^{\bot}\multimap B\vdash B^{\bot}\multimap A$ is not provable. Thus, we look at $A^{\bot}\multimap B$ as to a non-symmetric multiplicative disjunction, coherently with the non-symmetric look of the corresponding new syllogistic diagram $A\arr\b\arr B$. 
\end{remark}
\begin{remark}\label{primaopoi}
The SYLL diagram $B\arr\b\leftarrow\b\leftarrow A$, which is the reversal of the new syllogistic diagram for \aa{a}{b}, corresponds to the formula $B\multimap A^{\bot\bot}$. This makes sense because $A^{\bot}\multimap B^{\bot}\equiv B\multimap A^{\bot\bot}$. So, the occurrence of $A$ in the part $\b\leftarrow\b\leftarrow A$ can be actually considered as $A$ occurring in complemented mode twice, see example~\ref{pxx}. Similarly, the reversals of the new syllogistic diagrams for \ii{a}{B} and \ii{a}{b} are those for \ii{B}{a} and \ii{b}{a}.
\end{remark}
The next result extends theorem~\ref{perquesta}.
\begin{theorem}\label{pertutto}
A (strengthened) De Morgan syllogism is valid if and only if it is provable in $\rll^{\bot}$.
\end{theorem}
\begin{proof}
The proof is completely similar to the one for theorem~\ref{perquesta}.
\end{proof}
The next definition is given with respect to a suitable extension of definition~\ref{wfd} including the new syllogistic diagrams and their reversals among the available well-formed diagrams. The rule $(\ast)_A$ below is written in accordance with notation~\ref{parts}.
\begin{definition}\label{speriamo}
Let $\textup{SYLL}^{+*}$ denote the formal system which is obtained from $\textup{SYLL}^{+}$ by the addition of the rules
\[\begin{tabular}{ccccc}
\AxiomC{$A\arr\b\arr B$}\doubleLine\UnaryInfC{$B\leftarrow\b\leftarrow A$}\DisplayProof
&&\AxiomC{$A\arr\b\arr\b\leftarrow B$}\doubleLine\UnaryInfC{$B\arr \b\leftarrow \b\leftarrow A$}\DisplayProof&
\end{tabular}\]
\[\begin{tabular}{ccccc}
\AxiomC{$A\arr\b\leftarrow\b\arr B$}\doubleLine
\UnaryInfC{$B\leftarrow\b\arr\b\leftarrow A$}\DisplayProof
&&\AxiomC{$A\arr\b\leftarrow\b\arr\b\leftarrow B$}\doubleLine
\UnaryInfC{$B\arr\b\leftarrow\b\arr\b\leftarrow A$}\DisplayProof&
\end{tabular}\]
\vspace{.5cm}
\[\AxiomC{$\cate{D}\b\leftarrow A\arr\b\cate{E}$}
\RightLabel{$(\ast)_A$}
\UnaryInfC{\cate{D}A\cate{E}}
\DisplayProof\]
with suitable extension of the remaining notions introduced in definition~\ref{syllp}.
\end{definition}

\begin{remark}\label{uupp}
The rule $(\ast)_A$ in definition~\ref{speriamo} plays a r\^ole in the proofs 
of the (new) laws of the square of opposition described in proposition~\ref{nsq}, as sequents of $\textup{SYLL}^{+*}$. As an example, here immediately follows the abbreviated form of the proof of the (new) law of contradiction $A^{\bot}\multimap B, A^{\bot}\otimes B^{\bot}\vdash B\otimes B^{\bot}$, leaving the calculation of the remaining proofs to the reader: 
\[\AxiomC{$B\leftarrow\b\leftarrow A\arr\b\leftarrow\b\arr\b\leftarrow B$}
\RightLabel{$(\ast)_A$}
\UnaryInfC{$B\leftarrow A\leftarrow\b\arr\b\leftarrow B$}
\UnaryInfC{$B\leftarrow\b\arr\b\leftarrow B$}
\DisplayProof\]
More in general, the rule $(\ast)_A$ allows to cope with the valid De Morgan syllogisms whose middle term $M$ occurs in complemented in both the premises,
see table~(\ref{questat}), and then in a part like $\b\leftarrow M\arr\b$ in the diagram which is the concatenation of the diagrams for them, see definition~\ref{compl}. Consider $\aa{m}{P},\ii{S}{m}\therefore\ii{S}{P}$, for example. It corresponds to the sequent 
\begin{equation}\label{fi}
M\arr\b\arr P, S\leftarrow\b\arr\b\leftarrow M\models S\leftarrow\b\arr P
\end{equation}
in $\textup{SYLL}^{+*}$, a proof of which is
\[\AxiomC{$S\leftarrow\b\arr\b\leftarrow M$}
\AxiomC{$M\arr\b\arr P$}
\BinaryInfC{$S\leftarrow\b\arr\b\leftarrow M\arr\b\arr P$}
\RightLabel{$(\ast)_M$}
\UnaryInfC{$S\leftarrow\b\arr M\arr P$}
\UnaryInfC{\II{S}{P}}
\DisplayProof\]
or consider the strengthened De Morgan syllogism 
$\aa{m}{p},\aa{s}{m}, \ii{s}{s}\therefore\ii{s}{p}$. A proof of the corresponding sequent $\aa{m}{p},\aa{s}{m}, \ii{s}{s}\models\ii{s}{p}$ in $\textup{SYLL}^{+*}$
is
\[\AxiomC{}
\UnaryInfC{$S\arr\b\leftarrow\b\arr\b\leftarrow S$}
\AxiomC{$S\arr\b\arr\b\leftarrow M$}
\AxiomC{$M\arr\b\arr\b\leftarrow P$}
\BinaryInfC{$S\arr\b\arr\b\leftarrow M\arr\b\arr\b\leftarrow M$}
\BinaryInfC{$S\arr\b\leftarrow\b\arr\b\leftarrow S\arr\b\arr\b\leftarrow M\arr\b\arr\b\leftarrow P$}
\RightLabel{$(\ast)_S$}
\UnaryInfC{$S\arr\b\leftarrow\b\arr S\arr\b\leftarrow M\arr\b\arr\b\leftarrow P$}
\UnaryInfC{$S\arr\b\leftarrow\b\arr\b\leftarrow M\arr\b\arr\b\leftarrow P$}
\RightLabel{$(\ast)_M$}
\UnaryInfC{$S\arr\b\leftarrow\b\arr M\arr\b\leftarrow P$}
\UnaryInfC{$S\arr\b\leftarrow\b\arr\b\leftarrow P$}
\DisplayProof\]
\end{remark}
\begin{remark}\label{eco}
We observe that the rule $(\ast)_A$ involves the deletion of two bullet symbols. Thus, as it stands the rejection criterion described in remark~\ref{doppo} does not apply to the extended diagrammatic calculus in $\textup{SYLL}^{+*}$, since now it may be the case that the number of bullet symbols occurring in a (new) syllogistic diagram for a syllogistic conclusion is strictly less than their total number in the diagrams for the premises. For instance, on the base of such criterion, the syllogism $\aa{m}{P},\aa{S}{m}\models\aa{S}{P}$ would be wrongly rejected, despite the fact that it is provable in $\textup{SYLL}^{+*}$:
\[\AxiomC{$S\arr\b\leftarrow M$}
\AxiomC{$M\arr\b\arr P$}
\BinaryInfC{$S\arr\b\leftarrow M\arr\b\arr P$}
\RightLabel{$(\ast)_M$}
\UnaryInfC{$S\arr M\arr P$}
\UnaryInfC{$S\arr P$}
\DisplayProof\]
So, in $\textup{SYLL}^{+*}$ the rejection criterion described in remark~\ref{doppo} is still usable, provided considering possible deletions of bullet symbols due to the application of instances of the rule $(\ast)_A$.
\end{remark}

\begin{remark}
Despite the connection of the rule $(\ast)_A$ with the deletion of term-variables occurring in diagrams in complemented mode twice, as pointed out in remark~\ref{uupp}, we do not think that it should be related to an involutive notion of negation, which would be taken care by rules such as 
\[\begin{tabular}{ccccc}
&\AxiomC{$\cate{D}A\arr\b\arr\b\cate{E}$}
\UnaryInfC{$\cate{D}A\cate{E}$}
\DisplayProof&&
\AxiomC{$\cate{D}\b\leftarrow\b\leftarrow A\cate{E}$}
\UnaryInfC{$\cate{D}A\cate{E}$}
\DisplayProof&
\end{tabular}\]
The fact that the sequent $M^{\bot}\multimap P,S\otimes M^{\bot}\vdash S\otimes P$ in $\rll^{\bot}$, corresponding to the sequent~(\ref{fi}), is provable without any appeal to such a notion is in support of this.
\end{remark}
In order to extend theorems~\ref{propopropo} and~\ref{propos} to syllogistic with complemented terms, we proceed through the classification of the combinations of diagrammatic premises toward the conclusions of the valid syllogisms in table~(\ref{questat}), in the moods \catebf{AAA}, \catebf{AII}, \catebf{IAI}, \catebf{AAI}, correspondingly dividing the classification into four lemmas.
The pairs of premises listed in each lemma have been further divided accordingly to the occurrence of the middle term-variable $M$ in complemented mode in both or in neither of them, in accordance with the table~(\ref{questat}). 
For instance, this listing of the premises prevents the obtainment of the universal conclusions $S\arr\b\arr P$ and $S\arr\b\arr\b\leftarrow P$ from the pairs $(S\arr M, M\arr\b\arr P)$ and $(S\arr M, M\arr\b\arr\b\leftarrow P)$, respectively. We reject them on the base of the fact that in both, the term-variable $M$ occurs in complemented mode in their second components only. If possible, they would correspond to the (invalid) syllogisms $\aa{m}{P},\aa{S}{M}\therefore\aa{s}{P}$ and $\aa{m}{p},\aa{S}{M}\therefore\aa{s}{p}$ respectivley, and in turn to the (unprovable) sequents $S\multimap M, M^{\bot}\multimap P\vdash S^{\bot}\multimap P$ and $S\multimap M, M^{\bot}\multimap P^{\bot}\vdash S^{\bot}\multimap P^{\bot}$.
More to the point, we observe for example that the pairs of diagrammatic premises listed in the statement of lemma~\ref{aaa} below alternate so that in both their components the term-variable $M$ occurs in complemented mode in the even points (ii), (iv), (vi), (viii) only. A similar listing occurs in each of the remaining lemmas~\ref{aii},~\ref{iai} and~\ref{aai}.

The next lemma classifies the pairs of universal premises leading to a universal conclusion, namely those for the valid De Morgan syllogisms in the mood \catebf{AAA}.
\begin{lemma}\label{aaa}
The pairs of universal premises toward a universal conclusion, are exactly the following.
\begin{itemize}
\item[(i)] $(S\arr M, M\arr P)$
\item[(ii)] $(S\arr\b\leftarrow M,M\arr\b\arr P)$
\item[(iii)] $(S\arr M, M\arr\b\leftarrow P)$
\item[(iv)] $(S\arr\b\leftarrow M, M\arr\b\arr\b\leftarrow P)$
\item[(v)] $(S\arr\b\arr M,M\arr P)$
\item[(vi)] $(S\arr\b\arr\b\leftarrow M,M\arr\b\arr P)$
\item[(vii)] $(S\arr\b\arr M,M\arr\b\leftarrow P)$
\item[(viii)] $(S\arr\b\arr\b\leftarrow M, M\arr\b\arr\b\leftarrow P)$
\end{itemize}
\end{lemma}
\begin{proof}
It can be verified that each pair of universal premises in the statement leads to a universal conclusion in $\textup{SYLL}^{+*}$, through the employment of the rule $(\ast)_A$ in the cases (ii), (iv), (vi) and (viii). Conversely, we proceed by cases considering for the possible universal conclusions $S\arr P$, $S\arr\b\leftarrow P$, $S\arr\b\arr P$, $S\arr\b\arr\b\leftarrow P$ the ways to obtain each of them from a pair of universal premises with or without the intervention of the rule $(\ast)_A$.
\begin{itemize}
\item[(a)] on the base of remark~\ref{doppo}, the only way to obtain $S\arr P$ with no deletion of bullet symbols and the elimination of $M$ in non-complemented mode in both the premises is by (i), since in this case no bullet symbol is allowed to occur in the premises. The only way to obtain $S\arr P$ through the employment of the rule $(\ast)_A$ is by (ii), since in this case exactly two bullet symbols must occur in a pair of universal premises, one bullet per premise, considering the accordingly oriented arrow symbols in the parts $S\arr$, $\arr P$ of the first, resp. second, premise.
\item[(b)] on the base of remark~\ref{doppo}, the only way to obtain $S\arr\b\leftarrow P$ with no deletion of bullet symbols and the elimination of $M$ in non-complemented mode in both the premises is by (iii). The only way to obtain $S\arr\b\leftarrow P$ through the employment of the rule $(\ast)_A$ is by (iv), since in this case exactly three bullet symbols must occur in a pair of universal premises, two of them in the right-hand side premise, toward a universal conclusion with exactly one bullet symbol occurring, with two arrow symbols converging to it, that is by also considering the oppositely oriented arrow symbols in the parts $S\arr$, $\b\leftarrow P$ of the first, resp. second, premise.
\item[(c)] on the base of remark~\ref{doppo}, the only way to obtain $S\arr\b\arr P$ with no deletion of bullet symbols and the elimination of $M$ in non-complemented mode in both the premises is by (v). The only way to obtain $S\arr\b\arr P$ through the employment of the rule $(\ast)_A$ is by (vi), since in this case exactly three bullet symbols must occur in a pair of universal premises, two of them in the left-hand side premise, toward a universal conclusion with exactly one bullet symbol, together with an arrow symbol converging to it and an arrow symbol diverging from it, that is by also considering the accordingly arrow symbols in the parts $S\arr\b$, $\arr P$ in the first, resp. second, premise.
\item[(d)] on the base of remark~\ref{doppo}, the only way to obtain 
$S\arr\b\arr\b\leftarrow P$ with no deletion of bullet symbols and the elimination of $M$ in non-complemented mode in both the premises is by (vii). The only way to obtain $S\arr\b\arr\b\leftarrow P$ through the employment of the rule $(\ast)_A$ is by (viii), since in this case exactly four bullet symbols must occur in a pair of universal premises, two per premise, toward a universal conclusion with exactly two bullet symbols occurring, separating three arrow symbols as in the wanted conclusion, that is by also considering the parts $S\arr\b\arr$, $\arr\b\arr P$ in the first, resp. second, premise.
\end{itemize}
\end{proof}
The next lemma classifies the pairs of premises with particular first component and universal second component
leading to a particular conclusion, namely those for the valid De Morgan syllogisms in the mood \catebf{AII}.
\begin{lemma}~\label{aii}
The pairs of premises toward a particular conclusion whose first, resp. second, component is particular, resp. universal, are exactly the following:
\begin{itemize}
\item[(i)] $(S\leftarrow\b\arr M,M\arr P)$
\item[(ii)] $(S\leftarrow\b\arr\b\leftarrow M,M\arr\b\arr P)$
\item[(iii)] $(S\leftarrow \b\arr M,M\arr\b\leftarrow P)$
\item[(iv)] $(S\leftarrow\b\arr\b\leftarrow M,M\arr\b\arr\b\leftarrow P)$
\item[(v)] $(S\arr\b\leftarrow\b\arr M,M\arr P)$
\item[(vi)] $(S\arr\b\leftarrow\b\arr\b\leftarrow M,M\arr\b\arr P)$
\item[(vii)] $(S\arr\b\leftarrow\b\arr M, M\arr\b\leftarrow P)$
\item[(viii)] $(S\arr\b\leftarrow\b\arr\b\leftarrow M, M\arr\b\arr\b\leftarrow P)$
\end{itemize}
\end{lemma}
\begin{proof}
Similar to the proof of lemma~\ref{aaa}.
\end{proof}
The next lemma classifies the pairs of premises with universal first component and particular second component toward a particular conclusion, namely those for the valid De Morgan syllogisms in the mood \catebf{IAI}.
\begin{lemma}\label{iai}
The pairs of premises toward a particular conclusion whose first, resp. second, component is universal, resp. particular, are exactly the following:
\begin{itemize}
\item[(i)] $(S\leftarrow M, M\leftarrow\b\arr P)$
\item[(ii)] $(S\leftarrow\b\leftarrow M, M\arr\b\leftarrow\b\arr P)$
\item[(iii)] $(S\arr\b\leftarrow M,M\leftarrow\b\arr P)$
\item[(iv)] $(S\arr\b\leftarrow\b\leftarrow M, M\arr\b\leftarrow\b\arr P)$
\item[(v)] $(S\leftarrow M, M\leftarrow\b\arr\b\leftarrow P)$
\item[(vi)] $(S\leftarrow\b\leftarrow M,M\arr\b\leftarrow\b\arr\b\leftarrow P)$
\item[(vii)] $(S\arr\b\leftarrow M,M\leftarrow\b\arr\b\leftarrow P)$
\item[(viii)] $(S\arr\b\leftarrow\b\leftarrow M,M\arr\b\leftarrow\b\arr\b\leftarrow P)$
\end{itemize}
\end{lemma}
\begin{proof}
Similar to a proof of lemma~\ref{aii}, since the previous list has been obtained from the list in lemma~\ref{aii} by exchanging the premises and passing to a reversal.
\end{proof}
The next lemma classifies the $3$-tuples 
of premises made of an existential import and two universal components
that lead to a particular conclusion, namely those for the valid strengthened De Morgan syllogisms in the mood \catebf{AAII}.
\begin{lemma}~\label{aai}
The $3$-tuples of premises made of an existential import and two universal components that lead 
toward a particular conclusion, are exactly the following:
\begin{itemize}
\item[(i)] $(S\leftarrow\b\arr S,S\arr M, M\arr P)$
\item[(ii)] $(S\leftarrow\b\arr S, S\arr\b\leftarrow M,M\arr\b\arr P)$
\item[(iii)] $(S\leftarrow\b\arr S, S\arr M, M\arr\b\leftarrow P)$
\item[(iv)] $(S\leftarrow\b\arr S, S\arr\b\leftarrow M, M\arr\b\arr\b\leftarrow P)$
\item[(v)] $(S\arr\b\leftarrow\b\arr\b\leftarrow S, S\arr\b\arr M, M\arr P)$
\item[(vi)] $(S\arr\b\leftarrow\b\arr\b\leftarrow S, S\arr\b\arr\b\leftarrow M, M\arr\b\arr P)$
\item[(vii)] $(S\arr\b\leftarrow\b\arr\b\leftarrow S, S\arr\b\arr M, M\arr\b\leftarrow P)$
\item[(viii)] $(S\arr\b\leftarrow\b\arr\b\leftarrow S, S\arr\b\arr\b\leftarrow M, M\arr\b\arr\b\leftarrow P)$
\end{itemize}
\end{lemma}
\begin{proof}
See the proof of lemma~\ref{aaa}.
\end{proof}
The next theorem extends theorems~\ref{propopropo} and~\ref{propos}. We provide a full proof for the mood \catebf{AAA} only. The remaining cases are left to be proved to the reader in a completely analogous way. The proof is purely syntactical and based on the lemma~\ref{aaa}. On one hand we proceed top-down constructing a scheme of formal proof for any syllogism in the mood \catebf{AAA}, from the syllogistic diagrams for its premises. On the other hand we proceed bottom-up by cases, showing that the provable syllogisms in the mood \catebf{AAA} are among those of table~(\ref{questat}).
\begin{theorem}~\label{uuu}
A (strengthened) De Morgan syllogism is valid if and only if it is provable in $\textup{SYLL}^{+*}$.
\end{theorem}
\begin{proof}
On the base of lemma~\ref{aaa}, 
the syllogistic diagrams for the universal premises of a valid syllogism in the mood \catebf{AAA} from table~(\ref{questat}), 
form pairs $(S\cate{A}\arr M,M\arr\cate{B}P)$ or
$(S\cate{A}\arr\b\leftarrow M,M\arr\b\arr\cate{B}P)$. Lemma~\ref{aaa} ensures that the roots of the formal proofs
\[\AxiomC{$S\cate{A}\arr M$}
\AxiomC{$M\arr\cate{B}P$}
\BinaryInfC{$S\cate{A}\arr M\arr\cate{B}P$}
\UnaryInfC{$S\cate{A}\arr\cate{B}P$}
\DisplayProof
\qquad
\AxiomC{$S\cate{A}\arr\b\leftarrow M$}
\AxiomC{$M\arr\b\arr\cate{B}P$}
\BinaryInfC{$S\cate{A}\arr\b\leftarrow M\arr\b\arr\cate{B}P$}
\RightLabel{$(\ast)_M$}
\UnaryInfC{$S\cate{A}\arr M\arr\cate{B}P$}
\UnaryInfC{$S\cate{A}\arr\cate{B}P$}
\DisplayProof\]
are the syllogistic diagrams for the universal 
conclusion of any syllogism in the mood \catebf{AAA} from table~(\ref{questat}).\\

By (i) and (ii) of lemma~\ref{aaa}, the only ways to obtain $S\arr P$ as a conclusion of a formal proof are 
\[\AxiomC{$S\arr M\arr P$}
\UnaryInfC{$S\arr P$}
\DisplayProof\qquad
\AxiomC{$S\arr\b\leftarrow M\arr\b\arr P$}
\RightLabel{$(\ast)_M$}
\UnaryInfC{$S\arr M\arr P$}
\UnaryInfC{$S\arr P$}
\DisplayProof\]
that amount to the valid syllogisms $\aa{M}{P},\aa{S}{M}\therefore\aa{S}{P}$ and $\aa{m}{P},\aa{S}{m}\therefore\aa{S}{P}$, respectively.\\
By (iii) and (iv) of lemma~\ref{aaa}, the only ways to obtain $S\arr\b\leftarrow P$ as a conclusion of a formal proof are 
\[\AxiomC{$S\arr M\arr\b\leftarrow P$}
\UnaryInfC{$S\arr\b\leftarrow P$}
\DisplayProof\qquad
\AxiomC{$S\arr\b\leftarrow M\arr\b\arr\b\leftarrow P$}
\RightLabel{$(\ast)_M$}
\UnaryInfC{$S\arr M\arr \b\leftarrow P$}
\UnaryInfC{$S\arr \b\leftarrow P$}
\DisplayProof\]
that amount to the valid syllogisms $\aa{M}{p}, \aa{S}{M}\therefore\aa{S}{p}$ and $\aa{m}{p},\aa{S}{m}\therefore\aa{S}{p}$, respectively.\\
By (v) and (vi) of lemma~\ref{aaa}, the only ways to obtain $S\arr\b\arr P$ as a conclusion of a formal proof are
\[\AxiomC{$S\arr\b\arr M\arr P$}
\UnaryInfC{$S\arr\b\arr P$}
\DisplayProof
\qquad
\AxiomC{$S\arr\b\arr\b\leftarrow M\arr\b\arr P$}
\RightLabel{$(\ast)_M$}
\UnaryInfC{$S\arr\b\arr M\arr P$}
\UnaryInfC{$S\arr\b\arr P$}
\DisplayProof\]
that amount to the valid syllogisms $\aa{M}{P},\aa{s}{M}\therefore\aa{s}{P}$ and $\aa{m}{P}, \aa{s}{m}\therefore\aa{s}{P}$, respectively.\\ 
By (vii) and (viii) of lemma~\ref{aaa}, the only ways to obtain $S\arr\b\arr\b\leftarrow P$ as a conclusion ofa formal proof are
\[\AxiomC{$S\arr\b\arr M\arr\b\leftarrow P$}
\UnaryInfC{$S\arr\b\arr\b\leftarrow P$}
\DisplayProof\qquad
\AxiomC{$S\arr\b\arr\b\leftarrow M\arr\b\arr\b\leftarrow P$}
\RightLabel{$(\ast)_M$}
\UnaryInfC{$S\arr\b\arr M\arr\b\leftarrow P$}
\UnaryInfC{$S\arr\b\arr\b\leftarrow P$}
\DisplayProof\]
tha amount to valid syllogism $\aa{M}{p},\aa{s}{M}\therefore\aa{s}{p}$ and $\aa{m}{p},\aa{s}{m}\therefore\aa{s}{p}$, respectively. 
\end{proof}
The next theorem extends theorem~\ref{perquesta}.
\begin{theorem}\label{dsd}
A (strengthened) De morgan syllogism is provable in $\rll^{\bot}$ if and only if it is provable in $\textup{SYLL}^{+*}$.
\end{theorem}
\begin{proof}
The result follows from the theorems~\ref{pertutto} and~\ref{uuu}.
\end{proof}

%% file: final.tex
\section{Syllogisms and proof-nets}\label{comm}
Proof nets were introduced in~\cite{MR899269}. They are a graphical device that captures the essential geometric content of the proofs in the linear sequent calculus independently from their construction. The proofs of the same sequent that differ by the order of application of the rules of inference have the same proof-net. More to the point, proof-nets were introduced for the multiplicative fragment of classical linear logic, which we will describe below. Their introduction for the multiplicative intuitionistic fragment RLL, see definition~\ref{infrule}, is possible provided a suitable translation of this fragment into the classical one. Moreover, they need dedicated correcteness criteria. See~\cite{MR1714792, DanosRegnier, lamarche:inria-00347336}.

Focusing on the linear intuitionistic sequent calculus of syllogisms built on the (new) categorical formulas introduced in sections~\ref{syllrll} and~\ref{demorgan}, the aim of this section is that of arguing in favour of the employment of the diagrammatic sequent calculus of syllogisms developed in sections~\ref{sec4} and~\ref{demorgan} as a more direct diagrammatic proof method for the syllogistics in intuitionistic linear logic.

We took the syntax of classical linear logic and proof-nets from~\cite{MR899269, MR1356006, MR1003608}. 
\begin{definition}\label{cmll}
The system CMLL, for Classical Multiplicative Linear Logic, has 
formulas inductively constructed from atomic formulas $\alpha, \beta,\ldots,\alpha^{\bot},\beta^{\bot},\ldots$ in accordance with the grammar
$A\mathrel{\mathop :}=\bot\,|\,\catebf{1}\,|\,\alpha\,|\,\alpha^{\bot}\,|\,A\otimes A\,|\,A\parr A$, in which $(~)^{\bot}$ is \emph{linear negation}, $\otimes$ is \emph{multiplicative conjunction} and $\parr$ is \emph{multiplicative disjunction}. For every atomic formula $\alpha$, the atomic formula $\alpha^{\bot}$ is its linear negation and viceversa. In particular, $\bot^{\bot}=\catebf{1}$, $\catebf{1}^{\bot}=\bot$ and moreover, linear negation extends to composed formulas in accordance with the following De Morgan's laws:
\[(A\otimes B)^{\bot} \doteq A^{\bot}\parr B^{\bot}\qquad (A\parr B)^{\bot} \doteq A^{\bot}\otimes B^{\bot}\]
A sequent of CMLL is a formal expression $\Arr\Sigma$ where $\Sigma$ is a finite multiset of formulas.
The rules of CMLL are
\[\AxiomC{}
\RightLabel{(identity)}
\UnaryInfC{$\Arr A^{\bot},A$}
\DisplayProof\]
\vspace{.3cm}
\[\AxiomC{}
\RightLabel{(one)}
\UnaryInfC{$\Arr\catebf{1}$}
\DisplayProof
\hspace{3cm}
\AxiomC{$\Arr\Gamma$}
\RightLabel{(false)}
\UnaryInfC{$\Arr\Gamma,\bot$}
\DisplayProof\]
\vspace{.3cm}
\[\AxiomC{$\Arr\Gamma,A$}
\AxiomC{$\Arr\Delta, B$}
\RightLabel{(times)}
\BinaryInfC{$\Arr\Gamma,\Delta, A\otimes B$}
\DisplayProof
\hspace{3cm}
\AxiomC{$\Arr\Gamma,A,B$}
\RightLabel{(parr)}
\UnaryInfC{$\Arr\Gamma, A\parr B$}
\DisplayProof\]
A \emph{proof} in CMLL of a sequent $\Arr\Sigma$ is a finite tree whose vertices are sequents in CMLL, such that leaves are instances of the rule (identity), the root is $\Arr\Sigma$, and each branching is an instance of a rule of inference. A sequent of CMLL is \emph{provable} if there is a proof for it. Two formulas $A,B$ are said to be \emph{equivalent} if the sequents $\Arr A^{\bot},B$ and $\Arr B^{\bot},A$ are both provable, in which case we write $A\equiv B$.
\end{definition}
\begin{remark}
Typically, the sequents of classical linear logic are right-sided. This is just a 
notational convention which is adoptable thanks to the presence of linear negation, in virtue of which a two-sided sequent $A_1,\ldots,A_n\Arr B_1,\ldots,B_m$ can be equivalently replaced by a right-sided sequent $\Arr A_1^{\bot},\ldots,A_n^{\bot},B_1,\ldots,B_m$.
\end{remark}
\begin{proposition}\label{ggg}
For every formula $A$ of CMLL, the following facts hold:
\begin{itemize}
\item[(i)] $A\equiv A^{\bot\bot}$.
\item[(ii)] $1\otimes A\equiv A$.
\item[(iii)] $\bot\parr A \equiv A$.
\end{itemize}
\end{proposition}
\begin{proof}
Straightforward.
\end{proof}
\begin{notation}
In CMLL linear implication is defined. For every formulas $A,B$ of CMLL, put $A\multimap B\doteq A^{\bot}\parr B$. As a particular case of the rule (parr) in definition~\ref{cmll} one has the following derived rule
\[\AxiomC{$\Arr\Gamma,A^{\bot},B$}
\UnaryInfC{$\Arr\Gamma,A\multimap B$}
\DisplayProof\]
\end{notation}
The translation of the multiplicative intuitionistic fragment of linear logic into the classical one is in two pieces which are referred to as \emph{positive} and \emph{negative} with reference to the introduction of the so called \emph{polarized formulas}, positive and negative, which form a proper subset of the formulas of classical multiplicative linear logic. See~\cite{MR1714792}. 
The following definition introduces an adaptation of the positive-negative translation employed in loc. cit., suitable for the purposes of the present section.
\begin{definition}
The intuitionistic sytem $\rll^{\bot}$ translates into the classical system CMLL in accordance with the following positive and negative translations, \pos and \neg:
\[\left.\begin{array}{lllll}
\bot\pos\bot&&&&\bot\neg\catebf{1}\\
\alpha\pos\alpha&&&&\alpha\neg\alpha^{\bot}\\
A\otimes B\pos A\otimes B&&&&A\otimes B\neg B\multimap A^{\bot}\\
A\multimap B\pos A\multimap B&&&&A\multimap B\neg B^{\bot}\otimes A
\end{array}\right.\]
Every intuitionistic sequent $\Gamma\vdash A$ is translated into a classical sequent $\Arr\Gamma^{-},A^{+}$, where $\Gamma^{-}$ is the multiset of the classical formulas which are the negative translation of each intuitionistic formula in the multiset $\Gamma$ and $A^{+}$ is the classical formula which is the positive translation of the intuitionistic formula $A$.
\end{definition}

\begin{remark}
Because of points (ii) and (iii) in proposition~\ref{ggg}, using also that $\bot^{\bot}=\catebf{1}$ in CMLL,
it can verified that in particular $A^{\bot}\doteq A\multimap \bot$ in $\rll^{\bot}$, see notation~\ref{botbot}, translates in CMLL positively to $A^{\bot}$ and negatively to $A$.
\end{remark}
\begin{proposition}
If a sequent $\Gamma\vdash A$ is provable in $\rll^{\bot}$ then $\Arr\Gamma^-,A^+$ is provable in CMLL.
\end{proposition}
\begin{proof}
Straightforward.
\end{proof}
\begin{example}\label{fyl}
The syllogism $M\multimap P, S\multimap P\vdash S\multimap P$ in $\rll^{\bot}$ translates into the sequent $\Arr P^{\bot}\otimes M,M^{\bot}\otimes S,S\multimap P$. Here is their respective intuitionistic and classical proofs:
\begin{eqnarray}\label{lala}
\AxiomC{}
\UnaryInfC{$S\vdash S$}
\AxiomC{}
\UnaryInfC{$M\vdash M$}
\AxiomC{}
\UnaryInfC{$P\vdash P$}
\BinaryInfC{$M,M\multimap P\vdash P$}
\BinaryInfC{$M\multimap P, S\multimap M, S\vdash P$}
\UnaryInfC{$M\multimap P, S\multimap M\vdash S\multimap P$}
\DisplayProof
\end{eqnarray}
\begin{eqnarray}\label{pol}
\AxiomC{}
\UnaryInfC{$\Arr S^{\bot},S$}
\AxiomC{}
\UnaryInfC{$\Arr M^{\bot},M$}
\AxiomC{}
\UnaryInfC{$\Arr P^{\bot},P$}
\BinaryInfC{$\Arr P^{\bot}\otimes M,M^{\bot},P$}
\BinaryInfC{$\Arr P^{\bot}\otimes M,M^{\bot}\otimes S, S^{\bot},P$}
\UnaryInfC{$\Arr P^{\bot}\otimes M,M^{\bot}\otimes S, S\multimap P$}
\DisplayProof
\end{eqnarray}
A full translation of every syllogistic argument $P_1,P_2\therefore C$ into a  right-sided sequent in CMLL can be found in~\cite{Abrusci}.
\end{example}
\begin{definition}
A \emph{proof structure} in CMLL is a graph built from the following components:
\begin{itemize}
\item[-] \emph{link}: \[\begin{tikzpicture}
\draw (0,0) node [text=black] {$A$};
\draw (2,0) node [text=black] {$A^{\bot}$};
\draw (0,.2) -- (0,.5);
\draw (0,.5) -- (2,.5);
\draw (2,.2) -- (2,.5);
\end{tikzpicture}\]
\item[-] \emph{logical rules}:
\[\AxiomC{$A$}
\AxiomC{$B$}
\BinaryInfC{$A\otimes B$}
\DisplayProof
\hspace{3cm}
\AxiomC{$A$}
\AxiomC{$B$}
\BinaryInfC{$A\parr B$}
\DisplayProof\]
\vspace{.3cm}
\[\AxiomC{}
\UnaryInfC{$\catebf{1}$}
\DisplayProof
\hspace{3cm}
\AxiomC{}
\UnaryInfC{$\bot$}
\DisplayProof\]
with in particular
\[\AxiomC{$A^{\bot}$}
\AxiomC{$B$}
\BinaryInfC{$A\multimap B$}
\DisplayProof\]
\end{itemize}
Each formula must be the conclusion of exactly one rule and a premise of at most one rule. Formulas which are not premises are called \emph{conclusion of the proof structure}, which are not ordered.

A \emph{proof-net} is a proof structure which is constructed in accordance with the rules of the sequent calculus:
\begin{itemize}
\item[-] Links are proof nets.
\item[-] If $A$ is a conclusion of a proof-net $\nu$ and $B$ is a conclusion of a proof-net $\nu'$, then 
\[\AxiomC{$\xymatrix@R=4ex{\nu\ar@{--}[d]\\A}$}
\AxiomC{$\xymatrix@R=4ex{\nu'\ar@{--}[d]\\B}$}
\BinaryInfC{$A\otimes B$}
\DisplayProof\]
is a proof-net.
\item[-] If $A$ and $B$ are conclusions of the same proof net $\nu$, then 
\[\begin{tikzpicture}
\draw (0,.9) node {$\nu$};
\draw (0,0) node {\AxiomC{$\xymatrix@R=4ex{
\ar@{--}[d]&\ar@{--}[d]\\
A&B}$}
\UnaryInfC{$A\parr B$}
\DisplayProof};
\end{tikzpicture}\]
is a proof net.
\item[-] $\AxiomC{}\UnaryInfC{$\catebf{1}$}\DisplayProof$ is a proof net.
\item[-] If $\nu$ is a proof net, then
\[\AxiomC{$\nu$}
\UnaryInfC{$\bot$}
\DisplayProof\]
is a proof net.
\end{itemize}
\end{definition}
The following statement is the essential content of theorems 2.7 and 2.9 in~\cite{MR899269}.
\begin{theorem}
If $\pi$ is a proof in the multiplicative fragment of classical linear sequent  calculus without cut of $\Arr A_1,\dots, A_n$, then one can associate with $\pi$ a proof-net $\pi^*$ whose terminal formulas are exactly one occurrence of $A_1$,\ldots, one occurrence of $A_n$. Viceversa, if $\nu$ is a proof-net, one can find a proof $\pi$ in the sequent calculus such that $\nu=\pi^*$.
\end{theorem}
\begin{proof}
See~\cite{MR899269}.
\end{proof}
\begin{example}
The proof net
\[\begin{tikzpicture}
\draw (0,0) node {\AxiomC{$P^{\bot}$}
\AxiomC{$M$}
\BinaryInfC{$P^{\bot}\otimes M$}
\DisplayProof};
\draw (3,0) node {\AxiomC{$M^{\bot}$}
\AxiomC{$S$}
\BinaryInfC{$M^{\bot}\otimes S$}
\DisplayProof};
\draw (6,0) node {\AxiomC{$S^{\bot}$}
\AxiomC{$P$}
\BinaryInfC{$S\multimap P$}
\DisplayProof};

\draw (-.56,.35) -- (-.56,1.2);
\draw (-.56,1.2) -- (6.6,1.2);
\draw (6.6,1.2) -- (6.6,.35);

\draw (.56,.35)--(.56,.7);
\draw (.56, .7)--(2.44,.7);
\draw (2.44,.7)--(2.44,.35);

\draw (3.65,.35)--(3.65,.7);
\draw (3.65,.7)--(5.4,.7);
\draw (5.4,.7)--(5.4,.35);
\end{tikzpicture}\]
is the one associated with proof~(\ref{pol}) and captures its essential geometric content. It has to be observed that it is a \emph{planar graph}, namely one without crossing edges. In~\cite{Abrusci} it is proved that this is a general feature of the proof-nets for the syllogisms in the first figure, which characterizes them.  Actually, this is a purely geometrical justification of the reason why 
such syllogisms deserve a privileged status and were referred to as \emph{perfect} in section~\ref{redrules}. 
The proof-nets of the valid syllogisms in the remaining figures are necessarily not planar. For instance, the following is the proof-net associated with the syllogism $P\multimap M^{\bot},M\otimes S\vdash S\otimes P^{\bot}$ in the fourth figure, see table~(\ref{questa}), whose classical counterpart is $\Arr M\otimes P,S\multimap M^{\bot},S\otimes P^{\bot}$:
\[\begin{tikzpicture}
\draw (0,0) node {\AxiomC{$M$}
\AxiomC{$P$}
\BinaryInfC{$M\otimes P$}
\DisplayProof};
\draw (3,0) node {\AxiomC{$S^{\bot}$}
\AxiomC{$M^{\bot}$}
\BinaryInfC{$S\multimap M^{\bot}$}
\DisplayProof};
\draw (6,0) node {\AxiomC{$S$}
\AxiomC{$P^{\bot}$}
\BinaryInfC{$S\multimap P$}
\DisplayProof};

\draw (-.56,.35) -- (-.56,.7);
\draw (-.56,.7) -- (3.65,.7);
\draw (3.65,.35)--(3.65,.7);

\draw (2.25,.35)--(2.25,1);
\draw (2.25,1)--(5.4,1);
\draw (5.4,1)--(5.4,.35);

\draw (.56,.35)--(.56,1.3);
\draw(.56,1.3)--(6.4,1.3);
\draw(6.4,1.3)--(6.4,.35);
\end{tikzpicture}\]
In loc. cit., the reduction rules for the syllogistics are shown to be rules for the geometrical transformation of a non-planar proof-net into a planar one.

Now, in favour of the diagrammatic sequent calculus of syllogisms that we developed in the previous sections we can say that it applies to the intuitionistic sequents which express the syllogisms directly, rather than applying to a required translation of them in a classical framework. Because of this it maintains a close syntactical resemblance with the formulas occurring in such intuitionistic sequents. One can compare the intuitionistic and classical sequents for the syllogism in example~\ref{fyl}, for instance. Furthermore, our diagrammatic calculus captures the essential geometric content of the intuitionistic syllogistic proofs, getting rid of the order of application of the inference rules in them, like proof-nets do for the translated corresponding classical proofs. The essential geometric content of proof~(\ref{lala}) is captured by the single line of inference 
\[\AxiomC{$S\arr M\arr P$}
\UnaryInfC{$S\arr P$}
\DisplayProof\]
\end{example}
Finally, we introduced a diagrammatic calculus which allows to easily circumvent the writing of possibly long and tedious, although easy, proofs, simultaneously providing a criterion for their correctness, see theorems~\ref{ddd} and~\ref{dsd}, as well a criterion for their rejection, see remark~\ref{doppo}.

%% file: conclusions.tex
\section{Conclusions and further work}
In this paper we considered a natural reading of the traditional Aristotelian syllogistics within a well-known intuitionistic fragment of multiplicative propositional linear logic and put it in connection with the diagrammatic logical calculus of syllogisms that we introduced in~\cite{Pagnan2012-PAGADC}. Subsequently, we focused our interest on the nineteenth century modern De Morgan style syllogistics with complemented terms and have been able to extend the previously mentioned diagrammatic calculus to cope with it and with a linear logic reading of it, naturally extending the one for the traditional case. 

Hence, this paper presents itself as a contribution on an extended 
diagrammatic calculus for the syllogistic, an extended syllogistics in intuitionistic linear logic, and the existing connections between the two.

We believe that the existing connections between the diagrammatic formalism in this paper and the linear logic deserve a deeper investigation. A bidimensional extension of the former towards the representation of more general schemes of reasoning in linear logic, than the syllogistic ones, will constitute in the future the matter for further work.